\documentclass[lettersize,journal]{IEEEtran}
\usepackage{setspace} 
\ifCLASSOPTIONcompsoc
  \usepackage[nocompress]{cite}
\else
  \usepackage{cite}
\fi

\usepackage{float}
\usepackage{algorithm}
\usepackage{booktabs}
\usepackage{caption}
\usepackage{siunitx}

\usepackage{amsmath}
\usepackage{cleveref}
\usepackage{amsfonts}
\usepackage{algorithmic}
\usepackage{algorithm}
\usepackage{array}
\usepackage[caption=false,font=normalsize,labelfont=sf,textfont=sf]{subfig}
\usepackage{textcomp}
\usepackage{mathtools}
\usepackage{stfloats}
\usepackage{url}
\usepackage{verbatim}
\usepackage{graphicx}
\graphicspath{{./images/}}
\usepackage{subfig}                     
\usepackage{xr}

\usepackage{cite}
\usepackage{amssymb, amsbsy, amsthm, color, bm, bbm, ifthen}
\usepackage{balance}

\def\BibTeX{{\rm B\kern-.05em{\sc i\kern-.025em b}\kern-.08em
    T\kern-.1667em\lower.7ex\hbox{E}\kern-.125emX}}

\newtheorem{question}{\textbf{Question}}

\newtheorem{theorem}{Theorem}

\newtheorem{proposition}{Proposition}
\newtheorem{lemma}{Lemma}

\newtheorem{definition}{Definition}

\UseRawInputEncoding
\begin{document}

\title{Minimizing AoI in Mobile Edge Computing: Nested Index Policy with Preemptive and Non-preemptive Structure}

        





\author{Ning~Yang,
Yibo~Liu, 
Shuo~Chen,
Meng~Zhang*,
Haijun~Zhang, ~\IEEEmembership{Fellow,~IEEE,}

\thanks{
Ning Yang is with the Institute of Automation, Chinese Academy of Sciences, Beijing, 100190, China. (e-mail: ning.yang@ia.ac.cn).

Yibo Liu is with the Institute of Automation, Chinese Academy of Sciences, Beijing, 100190, China. 

Shuo Chen is with the Institute of Automation, Chinese Academy of Sciences, Beijing, 100190, China. 

Meng Zhang is with the ZJU-UIUC Institute, Zhejiang University, Zhejiang, 314499, China. (e-mail: mengzhang@intl.zju.edu.cn).

Haijun Zhang is with the University of Science and Technology, Beijing, China.

({*}Corresponding author:  Meng Zhang)

}}

\markboth{ }%
{Shell \MakeLowercase{\textit{et al.}}: Minimizing AoI in Mobile Edge Computing: Nested Index Policy with Preemptive and Non-preemptive Structure}

\maketitle

\begin{abstract}

\textit{Mobile Edge Computing} (MEC) leverages computational heterogeneity between mobile devices and edge nodes to enable real-time applications requiring high information freshness. The \textit{Age-of-Information} (AoI) metric serves as a crucial evaluator of information timeliness in such systems. Addressing AoI minimization in multi-user MEC environments presents significant challenges due to stochastic computing times. In this paper, we consider multiple users offloading tasks to heterogeneous edge servers in an MEC system, focusing on \textit{preemptive} and \textit{non-preemptive} task scheduling mechanisms. The problem is first reformulated as a \textit{Restless Multi-Arm Bandit} (RMAB) problem, with a multi-layer \textit{Markov Decision Process} (MDP) framework established to characterize AoI dynamics in the MEC system. 
Based on the multi-layer MDP, we propose a nested index framework and design a nested index policy with provably asymptotic optimality. This establishes a theoretical framework adaptable to various scheduling mechanisms, achieving efficient optimization through state stratification and index design in both preemptive and non-preemptive modes.
Finally, the closed-form of the nested index is derived, facilitating performance trade-offs between computational complexity and accuracy while ensuring the universal applicability of the nested index policy across both scheduling modes.
The experimental results show that in non-preemptive scheduling, compared with the benchmark method, the optimality gap is reduced by 25.43$\%$, while in preemptive scheduling, the gap has reduced by 61.84$\%$. As the system scale increases, it asymptotically converges in two scheduling modes and especially provides near-optimal performance in non-preemptive structure.

\end{abstract}

\begin{IEEEkeywords}
mobile edge computing (MEC), age of information (AoI), index policy, preemptive, non-preemptive
\end{IEEEkeywords}

%

\section{Introduction}

\subsection{Motivation}
Large-scale cyber-physical applications have an urgent demand for real-time information. The Internet of Things (IoT) devices are limited by their limited computing capabilities, so cloud computing is needed to improve performance. Meanwhile, vehicles generate a large amount of sensor data that needs to be processed, which is crucial to accurately present the surrounding environment and achieve navigation.
In the context of modern technological advancements, especially in the realm of real-time applications, users demand prompt status updates. The \textit{Age-of-Information} (AoI) is a recently introduced metric designed to assess the freshness of information, quantifying the time elapsed since the most recent message update (e.g., \cite{Tripathi2024AoI,Javani2025AoI}).


In numerous real-time applications, such as autonomous driving, the updated information is computationally demanding and necessitates processing. Offloading data to the cloud for computation can lead to data staleness and is computationally expensive. The \textit{Mobile Edge Computing} (MEC) paradigm shifts servers from the cloud to the edge, bringing users closer to servers and thereby reducing transmission delay (e.g., \cite{mi2022MEC, He2024MEC, Ning2021MEC}). Consequently, MEC emerges as a promising technology capable of reducing latency and enhancing information freshness. The majority of existing studies \cite{Zhu2024Optimizing, Li2025AoI,Yang2024Index,liu_indexability_2008,Xiao2025AoI,hsu_scheduling_2020,Song2024,liu_delay-optimal_2016} primarily concentrate on optimizing AoI in MEC systems with a single user or server, or under the assumption of fixed computation time and task size. 
However, in practical scenarios, multiple heterogeneous users and servers are prevalent, prompting further exploration of heterogeneous MEC systems. Nevertheless, minimizing AoI in MEC systems with heterogeneous servers presents two challenges: determining the optimal location for task offloading and deciding the time for this offloading. 
To this end, we first answer the following question:

\begin{question}
How should one minimize the AoI in a MEC system with multiple heterogeneous users and servers?
\end{question}

The task of minimizing AoI is frequently formulated as a \textit{Restless Multi-Arm Bandit} (RMAB) problem, as it can be optimally solved by value iteration \cite{gittins2011multi}. However, such strategies are prone to the curse of dimensionality, necessitating near-optimal solutions with low complexity. A promising method for addressing the RMAB problem is the index policy approach  \cite{zou2021minimizing}, which is particularly suitable for scheduling systems with multiple nodes. This approach can yield near-optimal results with relatively low computational complexity. The effectiveness of the index policy is attributed to two primary factors: the ability to decompose the original problem into several sub-problems with practicable optimal solutions and the potential to express the index in closed form for a specific \textit{Markov Decision Process} (MDP) structure, thereby reducing computational complexity. 
Regrettably, both factors are difficult to ensure under traditional MDP or \textit{Reinforcement Learning} (RL) models: the optimal solution of the subproblem may not exist, and obtaining the exponential function is not an easy task in a MEC system with heterogeneous users and servers due to the presence of multiple state variables. Traditional reinforcement learning algorithms are typically formulated for MDP with entirely unknown system dynamics, which often results in extended convergence periods \cite{He2024}. 

Furthermore, we also need to consider the MDP structure under preemptive and non-preemptive scheduling methods. Preemptive scheduling is typically applied in applications requiring dynamic resource adjustments, such as autonomous driving, where partial computation results can be discarded for fresher updates. Conversely, non-preemptive scheduling is indispensable in integrity-critical domains like medical imaging analysis, where interrupting task execution may lead to inconsistent results or operational hazards.
From the perspective of scheduling mechanisms, non-preemptive operation fundamentally represents a constrained variant of preemptive schemes under strict task interruption prohibitions. 
However, direct application of the preemptive strategy may interrupt the tasks and fail to make them complete.  For tasks with a completeness requirement, the preemptive optimized real-time resource reallocation strategy is no longer applicable and a redesign of the non-preemptive scheduling logic is required. 
But non-preemptive also brings new operational challenges: Non-preemptive scheduling mechanism brings different state transitions and constraints, leading to the need to design a threshold decision framework and a theoretical system for non-preemptive scheduling specifically.
For a task under non-preemptive scheduling, it cannot switch servers once the computation has started, preventing possible interruption of the task due to preemptive dynamic resource reallocation, which effectively ensures the completeness of the task. This constraint reformulation eliminates the feasibility of dynamically evaluating transient states to switch servers in preemptive systems, making server selection mandatory to be finalized at task initiation through a threshold-based decision framework.

Meanwhile, with more state variables, it is difficult to obtain a closed-form solution or a valid approximate solution to the problem by traditional methods, especially when decomposing the subproblems, and it is very difficult to obtain a closed-form index function. In addition, the direct application of MDP or RL methods often requires a large amount of computational resources, and the computational complexity rises dramatically when the system size increases. These mechanisms introduce different state transitions and decision constraints, increasing the difficulty of the problem. This leads us to the following question:

\begin{question}
How should we design an index-based policy for RMAB problems with multi-dimensional state variables, accounting for preemptive and non-preemptive scheduling dynamics?
\end{question}

These problems prompt us to turn to designing a method based on a multi-layer MDP model and combined with a nested index policy under both scheduling. In order to demonstrate the specificity of the non-preemptive mode comparing with preemption, we reveal the difference in state transfer under non-preemptive mode by proving the MLTT structure and fluid limit model under non-preemption. It shows that the nested index policy under non-preemptive scheduling still possesses asymptotic optimality, balancing the decision accuracy and computational complexity. And by decomposing the original problem into multiple sub-problems and reducing the computational complexity, it can better handle the complexity of the multi-dimensional state and decision constraints in the system while ensuring approximate optimality. 

\subsection{Solution Approach}
In response to this challenge, we propose a framework where multiple heterogeneous users offload tasks to heterogeneous edge servers under both preemptive and non-preemptive task scheduling mechanisms. And a multi-layer MDP model aimed at minimizing the average AoI has been constructed in MEC. By constructing a unified framework to model preemptive and non-preemptive constraints respectively in the multi-layer MDP, our nested index strategy demonstrates adaptability in different scheduling paradigms. The primary contributions of our research are as follows:

\begin{itemize} 
   \item[$\bullet$]\textit{Problem Formulation:}
    This study minimizes the average AoI in MEC systems by optimizing task offloading strategies. We convert the original issue into an RMAB problem and then use Lagrangian relaxation to break it down into sub-problems. It has been proven that a deterministic stationary policy represents the optimal solution for each of these sub-problems.

    \item[$\bullet$]\textit{Models for Different Scheduling:}
    A multi-layer MDP model  that simultaneously accommodates both preemptive and non-preemptive scheduling paradigms is constructed to solve the RMAB problem. Whether interruption is allowed or not, efficient optimization can be achieved through state stratification and index design, which proves the universality of the nested index policy.

    \item[$\bullet$]\textit{Nested Index Approach:}
     Leveraging multi-layer MDP, we rigorously proved the indexability of the multi-layer MDP in the MEC system and designed the corresponding index function and \textit{nested index} algorithm suitable for both preemptive and non-preemptive scheduling. \textit{To our knowledge, our nested index framework is the first model that effectively addressing the challenges of state transitions and decision-making constraints imposed by diverse users and edge servers across multiple states.}

    \item[$\bullet$]\textit{Numerical Results:}
     Compared to benchmarks, the nested index algorithm reduces the optimality gap by up to 25.43$\%$ in both preemptive and 61.84$\%$ in non-preemptive scenarios. As the system scale increases, it asymptotically converges in two scheduling modes and especially provides near-optimal performance in non-preemptive structure.
    
\end{itemize}

\section{Related Work}
To situate our work within the broader research landscape, this section reviews key advancements in three interrelated areas: AoI for information freshness, RMAB scheduling policies, and MEC resource management. By analyzing these domains, we identify critical gaps that our nested index framework aims to address.

\subsection{Age-of-Information}
The AoI metric has emerged as a pivotal framework for evaluating information freshness in single-node systems since it was proposed. Kaul \textit{et al.} in \cite{kaul2011minimizing} first proposed AoI as a metric to evaluate information freshness. The optimal AoI scheduling policy was to send messages from a source to the monitor through a single channel \cite{sun2017update}. In \cite{kadota_age_2021}, multiple sources could send updates over a single-hot network to a monitor, and they derived an approximate expression for the average AoI. Xie and Wang \cite{Xie2024} extended the AoI to Age of Usage Information (AoUI), considering more the usability of data. 
Some researchers will consider the optimization problem of AoI.
In \cite{Yates2019MultiSource}, they minimized AoI by considering multiple sources for queuing systems. 
 In \cite{Kadota2018Unreliablechannels}, they proposed a scheduling policy to minimize AoI in the wireless broadcast network with unreliable channels. In \cite{tripathi_whittle_2019}, they derived the structure of optimal policies for AoI minimizing problem and proved the optimality with reliable channel and unreliable channel assumptions. In \cite{Zhu2024Optimizing}, a fixed threshold policy was proposed in a non-preemptive system when  optimizing peak AoI. In \cite{Yu2023TWC}, Yu \textit{et al.} modeled the AoI minimization problem as a constrained MDP. E. Fountoulakis \textit{et al.} \cite{Foun2023} relaxed the constrained problem into an unconstrained MDP problem, and minimized the average AoI by using a backward dynamic programming algorithm with optimality guarantees.
 However, there was a lack of research on minimizing AoI in the more general MEC systems with heterogeneous multi-sources and edge servers under preemptive and non-preemptive scheduling.

\subsection{Restless Multi-Arm Bandit}

The RMAB problem arises when the state of an arm keeps changing whenever it is pulled or not \cite{whittle1988restless}. 
In \cite{liu_indexability_2008}, they formulated the problem of minimizing AoI as an RMAB problem and demonstrated that Whittle's Index was optimal when the arms were stochastically identical in a single-hop network. 
They also mentioned that a classic MDP was always indexable and proved the indexability of certain RMAB problems. 
Farag \cite{Farag2024} improved the AoI by proposing a relay selection mechanism within the Multi-Armed Bandit (MAB) framework.
Hsu \textit{et al.} \cite{hsu_age_2018} assumed that only one user could update at each time slot and obtained Whittle's Index in closed form. 
Hsu \textit{et al.} \cite{hsu_scheduling_2020} further studied the online and offline versions of the index approach and showed that the index policy was optimal when the arrival rate was constant. 
When there are multiple sensor-predictor pairs and multiple channels, M. K. C. Shisher \textit{et al.} in \cite{Shisher2023} turned the collaborative design problem into a multi-action RMAB.
In \cite{Chen2022UoI}, Chen \textit{et al.} developed a Whittle index policy that is near-optimal for the RMAB problem.
In \cite{Zhou2024AT}, Zhou \textit{et al.} proposed a sufficient condition for Whittle indexability based on the notion of active time (AT), and apply AT condition to the stochastic-arrival unreliable-channel AoI minimization problem. 
All the above index policies can only solve RMAB problems with one-dimensional state variables and lack the discussion on different scheduling manners. However, in general, there exist more factors that affect decisions in wireless networks. Therefore, we need to consider multiple state variables for general wireless networks under preemptive and non-preemptive scheduling.

\subsection{Mobile Edge Computing}

In MEC scenarios, mobile edge servers are well equipped with sufficient computational resources and are close to users, enabling them to expedite the computation process. Yang \textit{et al.} \cite{8892492,9322150}  studied the resource management problem in MEC utilizing reinforcement learning approaches. In \cite{liu_delay-optimal_2016}, an MDP-based policy was proposed to determine whether to offload a task and when to transmit it. Zou and Ozel in \cite{zou2021optimizing} studied the transmission and computation process for MEC systems as coupled two queues in tandem. The computing time is random in the MEC system. The optimal scheduling policy contains non-preemptive \cite{sun2017update} and preemptive \cite{zou2021optimizing} structures, respectively. In \cite{zou2021optimizing}, the optimal scheduling policy under preemptive structure had a threshold property, and it was a benefit for minimizing AoI to wait before offloading. 
For minimizing AoI problems with multiple sources (or users), they established the MDP model and index-based policies \cite{zou2021minimizing,liu2010indexability}, which had less complex and relatively efficient. Such an index-based policy was proved to be asymptotic for many single-hop wireless network scheduling.
Chiariotti \cite{Chiariotti2024} studied a theoretical model of an MEC server, deriving the expected AoI as well as the Peak AoI and delay distribution under different resource allocation strategies. In  \cite{He2024}, He \textit{et al.} constructed an online AoI minimization problem for the MEC system, introduced the Post-Decision State (PDS), and combined it with Reinforcement Learning.
In \cite{Song2024}, Song \textit{et al.}  constructed a multi-objective MDP model in the dynamic MEC system and proposed a multi-objective learning algorithm to minimize the total AoI. To summarize, random offloading time and indeterminate computation durations considering preemptive techniques posed significant offloading challenges in the MEC system.

\section{System Model}
To characterize the dynamic interactions between users, edge servers, and task scheduling paradigms in MEC systems, this section presents the system model. The system model provides the foundation for reformulating the AoI minimization problem as a RMAB problem with multi-layer MDP dynamics. 

\subsection{System Overview}
Consider an MEC system where a set $\mathcal{N}=\{1, 2, ..., N\}$ of users generate computational tasks and offload them to a set $\mathcal{M}= \{1, 2, ..., M\}$ of heterogeneous edge servers, as shown in Fig. \ref{fig:sketch}.
We let $n \in \mathcal{N}$ represent the index of users and $m \in\mathcal{M},\ $ denoting the index of edge servers. 

Let $t\in\mathcal{T}$ be the index of each time slot with $\mathcal{T} = \{ 1,2,...,T\}$. 
The generate-at-will model \cite{sun2017update,hsu_age_2018} is adopted, meaning that users can determine whether to generate a new task at each time slot $t$. In this MEC system, the transmission time between the user and the edge server is extremely short. Once the computation of one task is completed, its result is immediately sent back to the user. Each user can send a proportion of its task to any server for computing at each time slot \cite{partialoffloading}.

\begin{figure}[t]
    \centering
    \includegraphics[width=0.45\textwidth]{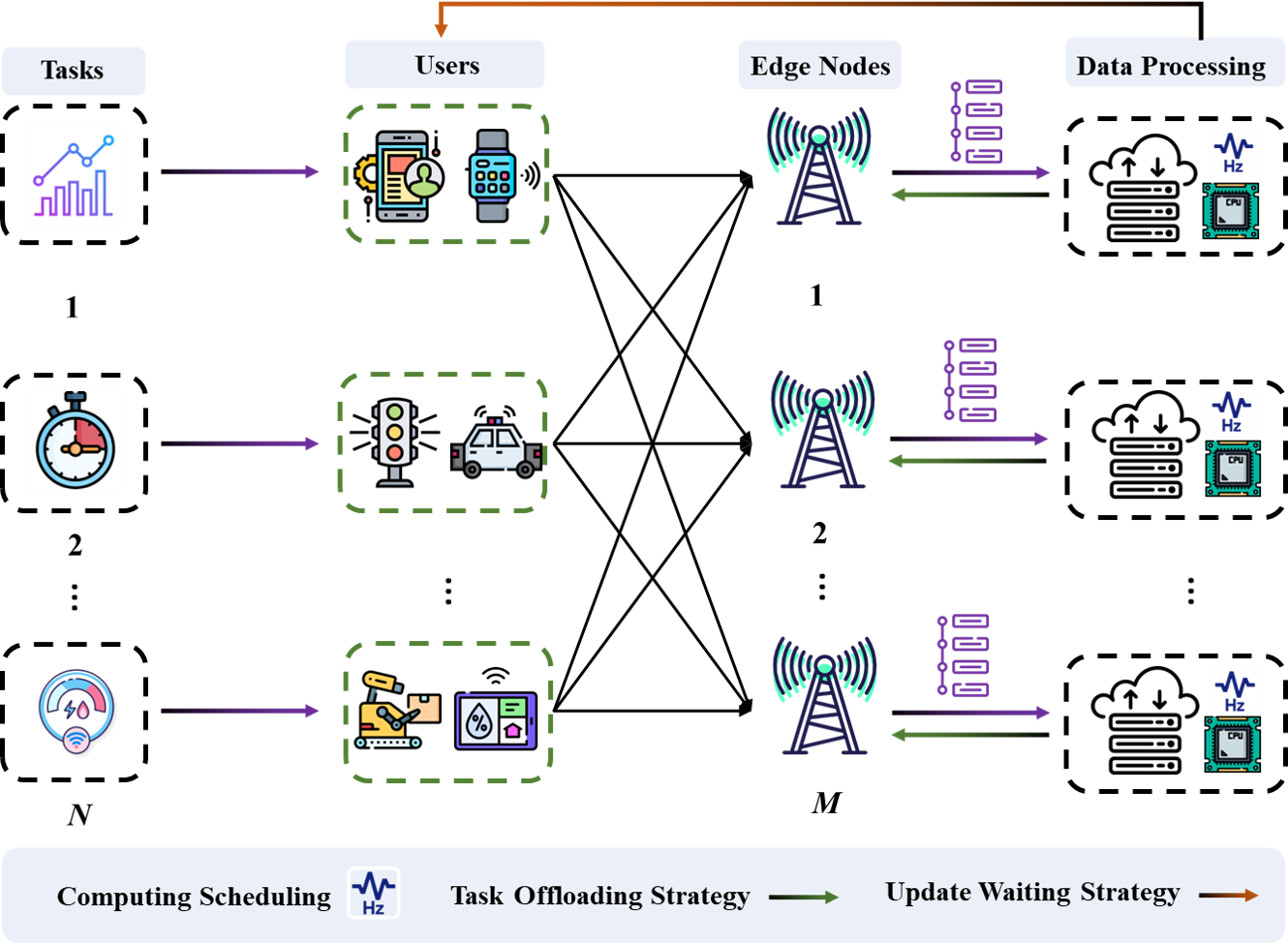}
    \caption{An MEC system in which each user offloads tasks to a certain edge server and receives data after processing.}
    \label{fig:sketch}
\end{figure}

\subsubsection{Task Scheduling Paradigms}

To characterize the dynamic interactions between task offloading decisions and server execution behavior, we explicitly model two critical scheduling paradigms:
\begin{itemize}
    \item \textbf{Preemptive scheduling}: Allows tasks to be interrupted and reassigned to other servers during computation, enabling adaptive resource allocation but introducing state transitions based on partial task completion.
    \item \textbf{Non-preemptive scheduling}: Requires tasks to complete execution on their initially selected server, ensuring computational integrity but limiting flexibility in response to changing conditions.
\end{itemize}

While preemptive scheduling supports real-time applications that require dynamic task reassignment to maintain information freshness, non-preemptive scheduling is important for ensuring predictable task processing sequences in scenarios that demand computational integrity.

These mechanisms fundamentally shape the system's dynamics and are incorporated into our model through the following assumptions: 
we assume that the edge servers are heterogeneous and the computing time of tasks is stochastic. Each task of user $n$ has a specific workload. When offloaded to a server, the task needs a specific number of CPU cycles to finish computing, and the computing time is based on both the workload and CPU frequency of the chosen server. Additionally, we posit that there is a minimum computing time $\tau_{n}^{min}$ for the task of user $n$.
We use AoI to measure the freshness of information. Specifically, we denote $\Delta_n(t)$ as the AoI of user $n$ at time $t$. 
The age of user $n$ decreases to the age of the latest off-loaded task when the computing is completed or increases by $1$ otherwise.
Let $G_n(t)$ denote the generation time of the most recent task offloaded by user $n$ at time $t$. Thus, when the computing is finished, the age of user $n$ at time $t$ can be expressed as
\begin{equation}\label{gnt}
       \Delta_n(t) = t-G_n(t), \quad \forall n \in \mathcal{N}.
\end{equation}



\subsubsection{Shifted Geometric Distribution}\label{probability}
For tasks offloaded to server $m$, the transition of AoI during computation follows a shifted geometric distribution \cite{taoexponential}. The parameter of this distribution is \(p_{m} = 1 - e^{-\lambda_{m}}\), where \(\lambda_{m}\) is the parameter of an exponential distribution.
To describe the task computing process more accurately, we set a minimum computing time \(\tau_{n}^{min }\) for each task. 
Specifically, during the initial \(\tau_n^{\text{min}}\) time slots following task initiation, edge servers are guaranteed to remain busy without task completion. Once this minimum computation duration \(\tau_n^{\text{min}}\) is exceeded, each server $m$ completes the task with probability \(p_m\) in every subsequent time slot. This modeling choice aligns with real-world computational constraints and, when combined with the shifted geometric distribution, jointly governs the AoI dynamics during task execution.

\subsection{Problem Formulation}

The freshness of information, quantified by AoI, is critical for real-time applications in MEC systems. To ensure optimal performance, under the preemptive and non-preemptive scheduling mechanisms, minimizing the overall AoI across heterogeneous users and servers becomes a central objective.
In the following, we formulate the AoI minimization problem.




\subsubsection{Constrained Long-Term Average AoI Optimization Model
 and Constraints}

Each user can choose one edge server to offload its tasks at time $t$. When a task is offloaded, the computation starts at the beginning of each time slot. 
We denote $y_{nm}(t)\in\{0,1\}$ as the offloading decision variable for user $n$ at time $t$: If user $n$ decides to offload a task to server $m$, then $y_{nm}(t)=1$. Otherwise, $y_{nm}(t)=0$, which means no task is offloaded. Under both preemptive and non-preemptive conditions, the decision conditions of $y_{nm}(t)$ are different. In preemptive scheduling, the user can interrupt the task being executed and reselect the server. Suppose the user $n$ uses the server $m$ at time $(t-1)$, i.e., $y_{nm}(t-1)=1$. When it switches to $m'\ (m'\neq m)$ at time $t$, there is  $y_{nm}(t-1)=0$ and $y_{nm'}(t)=1$. For non-preemptive scheduling, once it starts, until the task is completed, there is always $y_{nm}(t)=1$. 

Let $\pi\in\Pi_o$ denote the scheduling policy, which maps from the system state space $\mathcal{S}$ to the space of actions of all users. The set \(\Pi_o\) consists of all possible scheduling policies. The policy \(\pi\) is considered to be deterministic stationary \cite{sennott1989average}. The minimization problem of the long-term average AoI \cite{yates_age_2020, hsu_scheduling_2020,zou2021minimizing} under policy \(\pi\) is reformulated into the following form:
\begin{equation}
    \begin{aligned}
\label{reformulatedproblem}
\min_{\pi} & \quad \limsup_{T\to\infty}\frac{1}{TN}\sum_{t =1}^T\sum_{n\in\mathcal{N}}\mathbb{E}_{\boldsymbol{y}\sim\pi}[\Delta_n(t)], 
        \end{aligned} 
\end{equation}



\subsubsection{Subproblem Decomposition}

Based on the Lagrangian relaxation \cite{boyd2004convex}, we relax the instantaneous problem (\ref{reformulatedproblem}) to average constraint. 
It should be noted that the subproblem obtained by Lagrangian relaxation is not completely equivalent to the original problem \eqref{reformulatedproblem}, but rather an approximation. This is because there are differences in the feasible regions of the two objective functions. Denote  the set \(\Pi_n\) as the feasible region consisting of all possible scheduling policies for the relaxation problem. 
We relax the instantaneous constraint to an average constraint, which expands the feasible region, i.e., \(\Pi_o \subset \Pi_n\). Although the feasible regions are different, we will prove that the performance of the relaxation strategy asymptotically converges to the optimal of the original problem through the Fluid Limit Model. It allows for temporarily not satisfying the original constraint in some time slots, as long as the long-term average is satisfied.

Then we drop \eqref{reformulatedproblem} by introducing dual variables $\nu_m\ \ (\forall m\in\mathcal{M}),$ and deriving $N$ sub-problems. Specifically, this relaxation operation loosens the strict restrictions on each time slot, shifting from focusing on the specific offloading decisions for each time slot to considering the long-term average offloading situation, making the problem more solvable to a certain extent. Define $\pi_n\in\Pi_n$ as the policy that maps from the state of user $n$ to the action of user $n$. Given dual variables, each sub-problem $n$ is formulated as:

\begin{subequations}\label{originrelax}
\small\begin{align}
\min_{\pi_n} & \limsup_{T\to\infty} \frac{1}{T} \sum_{t =1}^T \mathbb{E}_{\bm{y}_n\sim\pi_n}
\Bigg[ \Delta_n(t) + \sum_{m\in\mathcal{M}} \nu_m y_{nm}(t) \Bigg] \label{eq:objective} \\
\text{s.t.} \quad & \sum_{m\in\mathcal{M}} y_{nm}(t) \leq 1, 
\    \ \forall n \in\mathcal{N}, \ t \in\mathcal{T}, \label{decoupled_problem} \\
& y_{nm}(t) \in \{0,1\}, 
\     \ \ \ \forall n \in\mathcal{N}, \ m \in\mathcal{M}, \ t \in\mathcal{T}. \label{eq:constraint2}
\end{align}
\end{subequations}
This framework applies to both preemptive and non-preemptive scheduling. In the subsequent chapters, we achieves differentiated optimization through distinct state transition models and index strategies tailored to each scheduling type.
When the dual variables converge, the sum of solutions to each sub-problem \eqref{originrelax} reaches the lower bound of the solution to problem \eqref{reformulatedproblem}.
To address the decomposed sub-problem (\ref{originrelax}) derived via Lagrangian relaxation, we leverage MDP to characterize the optimal policy structure. Specifically, in order to clarify that the deterministic stationary policy can solve the minimum long-term AoI of each sub-problem, we introduce Lemma \ref{detsta} according to the properties of MDP \cite{sennott1989average}, \cite{jiang2015approximate}:
\begin{lemma}\label{detsta}
    The optimal solution to each sub-problem \eqref{originrelax} is deterministic stationary.
\end{lemma}
The proof is deferred to Appendix A for clarity. This conclusion is applicable to both preemptive and non-preemptive scheduling modes. Although the state transitions and decision constraints in the two modes are different, the optimal policy structure of the subproblems has been proved to be deterministic and stationary.

\section{Analysis Framework for AoI Optimization}

Within the multi-layer MDP framework, the transition probabilities between system states under both scheduling mechanisms are critical to capturing the dynamic evolution of AoI. This section formalizes these probabilities for preemptive and non-preemptive task scheduling, which serve as the foundation for deriving the nested index policy in subsequent sections.

\subsection{Transition Probabilities}\label{defTP}

This section focuses on two offloading schemes: the preemptive manner and the non-preemptive manner. As shown in Fig. \ref{scheduling_flows}, the preemptive mechanism (Fig. \ref{preemptive_flow}) demonstrates dynamic task interruption capability when higher-priority tasks arrive, while the non-preemptive approach (Fig. \ref{nonpreemptive_flow}) maintains strict first-come-first-served execution order. Given the defined offloading decision and the shifted geometric distribution of computation time, we can write the transition probability of the AoI for the two conditions above:

\begin{figure}[!t]
    \centering
    \subfloat[Preemptive Scheduling Flow]{\includegraphics[width=1.85in]{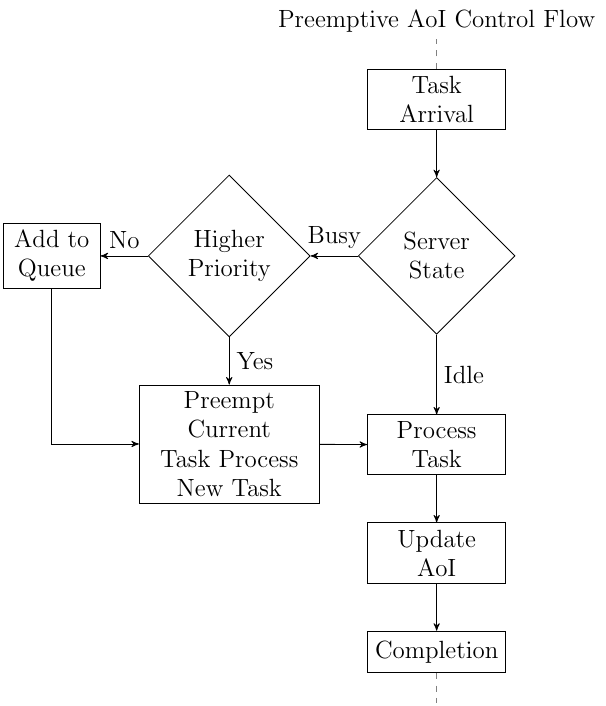}
    \label{preemptive_flow}}
    \hfil
    \subfloat[Non-preemptive Scheduling Flow]{\includegraphics[width=1.4in]{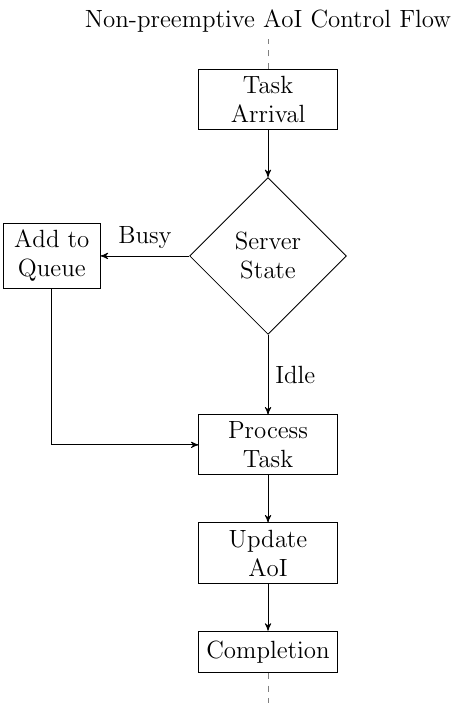}
\label{nonpreemptive_flow}}
\vspace{-1.2ex}
\caption{Task scheduling mechanisms comparison: (a) Preemptive mode allows task interruption with priority comparison, (b) Non-preemptive mode requires task completion before new assignments.}
\label{scheduling_flows}
\vspace{-1.8em}
\end{figure}

\subsubsection{Preemtive State Transition}

A preemptive manner allows a task at an edge server to either wait until it is completed or switch to a new task. In this condition, the transition probability for user $n$ during computing with a preemptive scheme can be represented as:
\begin{subequations}
\small\begin{align}
\mathbb{P}\{&\Delta_n(t + 1) = \Delta_n(t) + 1 \mid \nonumber \\& t-G_n(t) > \tau_n^{\min}, y_{nm}(t) = 1\} = 1-p_{m}, \label{eq:tran_prob_a}\\
\mathbb{P}\{&\Delta_n(t + 1) = t-G_n(t) + 1 \mid \nonumber\\&t-G_n(t) > \tau_n^{\min}, y_{nm}(t) = 1\} = p_{m}, \label{eq:tran_prob_b}\\
\mathbb{P}\{&\Delta_n(t + 1) = \Delta_n(t) + 1 \mid \nonumber \\& t-G_n(t) \leq \tau_n^{\min}, y_{nm}(t) = 1\} = 1, \label{eq:tran_prob_c}\\
\mathbb{P}\{&\Delta_n(t + 1) = \Delta_n(t) + 1 \mid y_{nm}(t) = 0\} = 1. \label{eq:tran_prob_d}
\end{align}
\end{subequations}
Eq.\eqref{eq:tran_prob_a} represents that when user $n$ chooses to offload the task to server $m$ at time $t$, i.e. $y_{nm}(t) = 1$, and the duration from the task generation time $G_{n}(t)$ to the current time $t$ is greater than the minimum computation time $\tau_{n}^{\min}$, the probability that the AoI of user $n$ increases by 1 at time $t + 1$ is $1-p_{m}$. Under this scenario, the task remains uncompleted with probability $1-p_m$, leading to an AoI increment of 1. Eq. \eqref{eq:tran_prob_b} demonstrates that when the elapsed time since task generation exceeds the minimum computation time \(\tau_n^{\text{min}}\), the task completes at time \(t+1\) with probability $p_{m}$. Eq. \eqref{eq:tran_prob_c} reveals that if this duration remains within \(\tau_n^{\text{min}}\) and the user offloads to serve $m$ at time $t$, the task necessarily remains uncompleted at \(t+1\), thus forcing the AoI of user $n$ to increase by 1. Eq. \eqref{eq:tran_prob_d} highlights that abstaining from offloading, i.e. \(y_{nm}(t)=0\), inherently leads to a unit increase in  AoI of user $n$ at \(t+1\).
Under preemptive scheduling, the priority of a task is mainly determined by whether certain conditions of the task exceed a specific threshold and the size of the task generation time, which will be explained in detail in the subsequent MLTT.

\subsubsection{Non-preemptive State Transition}
For tasks under non-preemptive, it is impossible to switch servers once the computation starts, preventing possible task interruptions due to preemptive dynamic resource reallocation. In this case, real-time resource reallocation policies optimized for preemptive scenarios become inapplicable, and thus the scheduling logic needs to be redesigned.
Therefore, when $t-G_n(t) > \tau_n^{\min}$, we can represent the transition probability as:
\begin{subequations}\label{probability_group}
\begin{align}
\mathbb{P}\{&\Delta_{n}(t+1) = \Delta_{n}(t) + 1 \mid \nonumber\\&b_m(t) = 0, y_{nm}(t) = 1 \} = 1 - p_m, \label{prob1} \\
\mathbb{P}\{&\Delta_{n}(t+1) = t - G_n(t) + 1 \mid\nonumber \\& b_m(t) = 0, y_{nm}(t) = 1 \} = p_m, \label{prob2} \\
\mathbb{P}\{&\Delta_{n}(t+1) = \Delta_{n}(t) + 1 \mid \nonumber \\& b_m(t) = 1, y_{nm}(t) = 1\} = 1. \label{prob3}
\end{align}
\end{subequations}
where $b_m(t)$ represents the idle state of the server. If $b_m(t)$ = 1, it indicates that server $m$ is busy and processing a task at time $t$, while  $b_m(t)$ = 0 indicates that server $m$ is idle. Specifically, Eq. (\ref{prob1}) and Eq. (\ref{prob2}) are similar to Eq. (\ref{eq:tran_prob_a}) and Eq. (\ref{eq:tran_prob_b}) in the preemptive mode. The difference is that there is an additional condition \(y_{nm}(t-1)=1\), which requires that user \(n\) offloads the task to server \(m\) at both time \(t\) and \(t-1\). In the non-preemptive mode, once the server starts the execution process, it cannot switch servers. Therefore, only when the task is offloaded to the same server for two consecutive time instants and the duration from the task generation time to the current time is greater than the minimum computation time, the AoI will change accordingly depending on whether the task is completed or not, corresponding to probability \(p_{m}\) and probability \(1-p_{m}\) respectively. Eq. (\ref{prob3}) indicates that in the non-preemptive mode, if user \(n\) selects server \(m\) at time \(t\) and server \(m'\) (\(m\neq m'\)) at time \(t-1\), then the AoI of user \(n\) will increase by 1 at time \(t + 1\). This is because the task cannot switch servers during the computation, and in this case, the task will be discarded. For the non-preemptive state transition, the task removes the constraint of ``urgent tasks interrupt low priority tasks” in the preemptive style, forcing the non-preemptive mode to seek a new balance between ``too early start leads to wasted server idle” and ``too late start leads to AoI accumulation''. The asymptotic optimality proofs and other theoretical frameworks applicable to preemption need to be reconstructed.

\subsection{Multi-Layer MDP}

We decomposed the original problem into sub-problem (\ref{originrelax}) using Lagrangian relaxation method. However, solving this sub-problem to derive the optimal solution under both preemptive and non-preemptive scheduling mechanisms remains challenging.
Given the deterministic stationary property established in Lemma \ref{detsta}, a multi-layer MDP framework is proposed to characterize the optimal policy \(\pi_{n}^{*}, \forall n\in\mathcal{N}\) for each decomposed sub-problem \ref{originrelax} under given dual variables \(\{\nu_{1}, \nu_{2}, \ldots, \nu_{M}\}\). This MDP model captures the dynamic evolution of system states and the impact of user offloading decisions, offering a novel approach to solving sub-problem (\ref{originrelax}). Specifically, the model integrates both preemptive and non-preemptive scheduling mechanisms by explicitly modeling their state transitions and decision constraints.

\begin{definition}[$L$-Layer MDP]\label{hmdp}
An $L$-layer MDP is a tuple $\langle\mathcal{S},\mathcal{A}, \mathcal{P}, \mathcal{C},L\rangle$, where $\mathcal{S}$ denotes the state space, $\mathcal{A}$ is the action space, $PD(\mathcal{S})$ is the probability distribution on the state space.
The transition function is $\mathcal{P}:\mathcal{S}\times \mathcal{A} \to PD(\mathcal{S})$, the cost function is $\mathcal{C}:\mathcal{S}\times \mathcal{A} \to \mathbb{R}$, and $L\in\mathbb{Z}_+$ is the number of layers. Denote $\mathcal{S}_l=\mathbb{N}^l$ as the state space at layer $l$, and $\mathcal{S}=\{\mathcal{S}_1,\mathcal{S}_2, \cdots,\mathcal{S}_L\}$. An $L$-layer MDP fulfills the following conditions: $\mathcal{S}_l\subset\mathbb{N}^L$, and $\mathcal{S}=\{\mathcal{S}_1,\mathcal{S}_2, \cdots,\mathcal{S}_L\}$
\begin{itemize}
    \item For all $0<l<L$ and all $s\in\mathcal{S}_l$, there exist an $a\in\mathcal{A}$ and $s'\in\mathcal{S}_{l+1}$ such that $\mathbb{P}\{s' \mid s, a\}>0$,

    \item For all $0<l\le L$ and all $s\in\mathcal{S}_l$, there exist an $a\in\mathcal{A}$ and $s'\in\mathcal{S}_{l}$ such that $\mathbb{P}\{s'\mid s, a\}>0$,

    \item For $1<l \le L$, there exists an $s\in\mathcal{S}_l$, $s'\in\mathcal{S}_1$, and $a\in\mathcal{A}$ such that $\mathbb{P}\{s'\mid s, a\}>0$,
\end{itemize}
and we term the sub-space $\mathcal{S}_l$ as layer $l$.
\end{definition}

The multi-layer MDP ensures state communication between different layers. The state at layer $l$ is able to transit to states at layer $l$, $l+1$, and layer $1$. 

We next specify the 2-layer MDP for the MEC system. In the single-layer MDP, the transition between states is usually determined by a single threshold. That is, when a state transitions to another state, the decision is made only relying on this threshold. 
However, the MEC system scenario in this study involves multiple different states, not a simple binary state switch. 
Specifically, there are multiple states in the system, such as tasks not being computed, tasks being computed and waiting for results, and tasks being completed. The transition and decision of these states cannot be determined by only one threshold. Therefore, we construct each sub-problem (\ref{originrelax}) as a 2-layer MDP model, which can set different levels of thresholds for the decision of different state transitions, to more accurately simulate and optimize the process of task offloading and information update in the MEC system:

\begin{itemize}
    \item \textbf{Action space:} 
    Let $\mathcal{A}$ be the action space for each user.
    For user $n$, its action is to select a server $m\in\mathcal{M}\cup\{0\}$, where $m=0$ means not to offload the task. Each action corresponds to a binary vector $\bm{y}_n(t)=\{y_{n1}(t),y_{n2}(t),\dots,y_{nM}(t)\}\in\mathcal{A}$, satisfying \(\sum_{m \in \mathcal{M}} y_{nm}(t) \leq 1\)and \( y_{nm}(t) \in \{0, 1\} \). That is, each user can choose at most one server or not choose any. 
    The set of actions of all users needs to satisfy the server capacity: \(\sum_{n \in \mathcal{N}} \sum_{m \in \mathcal{M}} y_{nm}(t) \leq M, \forall t\).
    The global action space $A$ is the combination of all possible \(\{y_{nm}(t)\}_{n,m}\), satisfying the above three constraints:

    \begin{equation}\label{action}
        \mathcal{A} = \left\{ 
\bm{y}_{n}(t),{n \in \mathcal{N}} \,\middle|\, 
\begin{aligned} 
& \sum_{n \in \mathcal{N}} \sum_{m \in \mathcal{M}} y_{nm}(t) \leq M, \forall t\\
& \sum_{m \in \mathcal{M}} y_{nm}(t) \leq 1, \forall n, t \\
& y_{nm}(t) \in \{0, 1\}, \forall n, m, t 
\end{aligned}
\right\}
    \end{equation}


    \item \textbf{State space:} 
    The state space for each user is partitioned into two layers:
    
    \textbf{1) \textit{Layer 1:}} All users are not in the computing task execution state. Characteristics are described by the current AoI $\Delta_n(t)$, forming the state space: 
    \begin{equation}
        S_1 = \{\Delta_n(t)|\Delta_n(t) \in \mathbb{N}\}
    \end{equation}
    
     \textbf{2) \textit{Layer 2:}} Users actively executing computational tasks. The state space in Layer 2 is defined as      
     \begin{equation}
         S_2 = \{(\Delta_n(t), D_n(t)) \} 
     \end{equation}where $D_n(t) = \Delta_n(G_n(t))$ represents the AoI at the task generation time $G_n(t)$.
    \item \textbf{Transition function:}
The state transition dynamics under both scheduling paradigms are modeled using \(\Delta_n(t)\). However, solely using $\Delta_n(t)$ can not fully characterize the transition of the states in multi-layer MDP. The transition probability from state $s$ to \(s'\) when user $n$ selects server $m$ is denoted as \(q_{nm}^{ss'}\), and we have
\begin{equation}
    q_{nm}^{ss'}=\mathbb{P}\{s'\ |\ s,\ y_{nm}(t)=1\},
\end{equation}
where $q_{nm}^{ss'}$ can be derived from the transition probability in Section \ref{defTP}. The transition function for this formation can include both preemptive and non-preemptive manners. 
Let $l \in \{1,2\}$ denote the layer where the state s is located.
To be specific, we have $q_{nm}^{ss'}=1$ for the following conditions:

\begin{subequations}\label{threecondition}
\setlength{\jot}{1pt} 
\addtolength{\abovedisplayskip}{-2pt} 
\small\begin{align} 
&q_{nm}^{ss'} = \mathbb{P}\Bigl\{ 
    s' = \Delta_n(t) + 1 \Bigm| l=1, \nonumber \\[-5pt]
    &\phantom{{}= \mathbb{P}\Biggl\{}s = \Delta_n(t),\ y_n(t) = 0 
\Bigr\} = 1, \label{threecondition1} \\[-2pt]
&q_{nm}^{ss'} = \mathbb{P}\Biggl\{ 
    s' = \Bigl( \Delta_n(t) + 1,\, \Delta_n(t) \Bigr) 
    \Bigm| l=1, \nonumber \\[-5pt]
    & \phantom{{}= \mathbb{P}\Biggl\{} s = \Delta_n(t),\ y_{nm}(t) = 1 
\Biggr\} = 1, \label{threecondition2} \\[-5pt]
&q_{nm}^{ss'} = \mathbb{P}\Biggl\{ 
    s' = \Bigl( \Delta_n(t) + 1,\, D_n(t) \Bigr) 
    \Bigm| l=2, \nonumber \\[-5pt]
    &s = \Bigl( \Delta_n(t),\, D_n(t) \Bigr), 
        \Delta_n(t) - D_n(t) \leq \tau_n^{\mathrm{min}}
\Biggr\} = 1.\label{threecondition3}
\end{align}
\end{subequations}

The above three formulas correspond to three situations that 1) When user $n$ does not offload, i.e., the layer $l=1$ and \(y_n(t)=0\), the age increases by 1, transitioning from \(s = \Delta_n(t)\) to \(s' = \Delta_n(t)+1\); 2) Offloading to server $m$, i.e., \(y_{nm}(t)=1\), shifts the state from \(s = \Delta_n(t)\) to \(s' = (\Delta_n(t)+1, \Delta_n(t))\), where \(\Delta_n(t)+1\) is the new age and \(\Delta_n(t)\) is the generation time \(D_n(t)\). 3) when the difference between the AoI of user \(n\) and the age at the time of generating the latest task is less than the minimum computation time, the state transition will necessarily occur.

\item \textbf{Cost function:} We define the immediate cost as 
\begin{align}\label{costfunction}
    C_n(s_n(t),m)\triangleq \Delta_n(t) + \nu_m ,
\end{align}
\noindent  which includes the current AoI and the server cost. The generating age $D_n(t)$ does not affect the immediate cost. The choice of actions determines the specific selection of server $m$, thereby ultimately affecting the server cost $\nu_m$

\end{itemize}

According to Lemma \ref{detsta}, there exists one deterministic stationary policy $\pi_n^*$ that reaches optimal average AoI for the relaxed sub-problem \eqref{originrelax}. However, value iteration when deriving the optimal policy for the original problem suffers from the curse of dimensionality \cite{optimalcontrol}. Specifically, the state space of the MEC system expands exponentially with the number of state variables, leading to prohibitive computational complexity and infeasible runtime.
Therefore, we seek for an approach that has less complexity and
achieves near optimality.

\section{Index-Based Policy}

In this section, a nested index approach is introduced to address the RMAB problem, with a particular focus on its application in both preemptive and non-preemptive scheduling within MEC systems. This approach has been proven to achieve asymptotic optimality in task offloading decisions, regradless of whether the tasks are scheduled preemptively or non-preemptively. We first formalize the nested index concept and establish the indexability of the proposed 2-layer MDP model for MEC systems. Next, a nested index policy is proposed for task scheduling.  It is explicitly designed to balance the flexibility that preemptive scheduling offers and the computational integrity of non-preemptive scheduling. Furthermore, the asymptotic optimality of the proposed approach is validated, and the nested index function is obtained in a closed form. This optimality holds true for both two scheduling within the MEC system. To accommodate both scheduling paradigms, the nested index framework is extended via structural modifications to the MDP layers.

\subsection{Core Concept of Passive Set and Intra-Indexability}

In the multi-layer MDP system, to better evaluate the pros and cons of different server selections under both both preemptive and non-preemptive scheduling, the concepts of the passive set and intra-indexability are introduced. 
First, we define the expected cost of choosing server $m$ given state $s$ as
\begin{equation}\label{bellman}
    \mu_{nm}(s,\bm{\nu})\triangleq C_n(s,m)+\sum_{{s}'\in\mathcal{S}}q^{ss'}_{nm}V_n({s}',\bm{\nu})
\end{equation}
where  function $V_n(s, \bm{\nu})$ is the differential cost-to-go \cite{optimalcontrol}. The cost-to-go function \(V_n(s, \nu)\) is used to evaluate the value of state $s$ under the given cost $\bm{\nu}$.   $\sum_{s'\in\mathcal{S}_l}q^{ss'}_{nm}V_n(s',\bm{\nu})$ reflects the impact of future states on the current decision-making cost. 
Following \cite{whittle1988restless}, the \textit{passive set} is defined with the dual variables in (\ref{originrelax})  as the activating cost vector \(\nu \triangleq (\nu_1, \ldots, \nu_M)\), where \(\nu\) denotes the computation costs.
\begin{definition}[Passive Set]
The passive set for user $n$ to transit to layer $l$ at server $m$ given activating cost $\bm{\nu}$ is denoted as:
\begin{equation}
\begin{aligned}
       &\mathcal{P}_{nm}^l(\bm{\nu})\triangleq \\&\left\{s\in\mathcal{S}_l\ |\min_{m'\in\mathcal{M}, m'\neq m }\mu_{nm'}(s,\bm{\nu})\leq \mu_{nm}(s,\bm{\nu})\right\}.
    \label{eq:passiveset}
\end{aligned}
\end{equation}
 We denote $\mathcal{P}_{nm}(\bm{\nu})\triangleq\cup_{l=1}^L\mathcal{P}_{nm}^l(\bm{\nu})$ as the overall passive set.
\end{definition}
The passive set \(\mathcal{P}_{n m}^{l}(\nu)\) for user $n$ at layer $l$ is defined as a subset of states \(S_l\). Specifically, for at least one alternative server \(m' \neq m\), the expected cost of choosing \(m'\) is less than or equal to that of selecting server $m$. That is, it refers to the set of states at layer $l$ that are sub-optimal for selecting server $m$ for computing with activating costs $\bm{\nu}$. In classic RMAB problems \cite{hsu_age_2018, kriouile_global_2021}, activating cost $\nu$ is a scalar. Whittle \cite{whittle1988restless} established that a problem is \textit{indexable} if the size of the passive set increases monotonically from $0$ to $+\infty$ as activating cost $\nu$ increases from $0$ to $+\infty$. 
Each state \(s \in S_l\) is assigned a critical activation activating cost \(\nu^*(s)\), defined as the maximum \(\nu\) where the optimal action remains unchanged, i.e., $\nu^*(s)=\max \left\{ \nu \mid \mu_{n m}(s, \nu) \leq \mu_{n m'}(s, \nu),\forall m' \neq m \right\}$.
The activating cost also gives the urgency of such a state as Whittle stated \cite{whittle1988restless}.
A state with a higher \(\nu^*(s)\) has a higher priority for selection. 

In MEC systems, \textit{indexability} is essential for designing efficient task offloading policies that balance user states and server costs. It refers to the property that for each layer in the multi-layer MDP. As the cost \(\nu_m\) increases, the size of the passive set grows monotonically. This monotonic growth ensures that each state has a unique critical cost threshold, allowing us to determine the optimal server by comparing the current cost against these thresholds. 
However, the indexability analysis becomes more complex due to two factors: \textit{Multi-action choices} and \textit{Heterogeneous servers}. Specifically, each user can choose from multiple servers and each server has distinct activation costs \(\nu_m\), which complicate the definition of indexability:
\begin{definition}[Intra-Indexability]\label{Def4}
\textit{In a Multi-Layer MDP} $\langle{\mathcal{S},\mathcal{A}, \mathcal{P}, \mathcal{C},L}\rangle$, given servers cost $\bm{\nu}$, if for any layer $l$, the cardinality of passive set $|\mathcal{P}_{nm}^l(\bm{\nu})|$ increases monotonically to the cardinality $|\mathcal{S}_l|$ of layer $l$ as cost $\nu_m$ for server $m$ increases from $0$ to $+\infty$, then this multi-layer MDP is intra-indexable.
\end{definition}

Using intra-indexability, we analyze how server costs relate to optimal server selection at each layer. By leveraging the monotonic increase in passive set size as costs rise, we determine scheduling thresholds to decide when to offload tasks to each server. This allows us to dynamically optimize task offloading policies, minimizing AoI in the system.
Given layer $l$, there exists the largest server cost $\nu_m'$. Beyond this threshold, state $s_n(t)$ is no longer included in the passive set $\mathcal{P}^l_{nm}(\bm{\nu})$. The monotonic property ensures that this critical cost is uniquely determined. Moreover, $|\mathcal{P}_{nm}(\bm{\nu})|$ is non-decreasing in $\bm{\nu}$ if $|\mathcal{P}_{nm}^l(\bm{\nu})|$ is non-decreasing in $\bm{\nu}$, for each $1\le l\le L$.

\subsection{Nested Index Policy Under Preemptive Structure}
In the research of the MEC system, to solve the optimization problem of sub-problem \eqref{originrelax}, we introduce the Bellman equation. 

\subsubsection{The Bellman Equation and MLTT Structure of Preemptive Scheduling}
The Bellman equation for each sub-problem  (\ref{decoupled_problem}) can be formulated as:
\begin{equation}
\begin{aligned}
        &\gamma_n^*+V_n(s, \bm{\nu})= \\&\min\limits_{m\in\mathcal{M}, l\in \mathcal{L}}\left[ C_n(s,m)+\sum\limits_{s'\in\mathcal{S}_l}q^{ss'}_{nm}V_n(s',\bm{\nu})\right]
\end{aligned}\label{BE}
\end{equation}
where \(\gamma_n^*\) is the optimal average cost per stage. The right side of the equation indicates that for all possible servers $m$ and all possible layers $l$, decision-making choices are made to find the server and layer that minimize the overall cost.

This Bellman equation can be adjusted in real-time according to the changes of the system state, which  fits well with the preemptive condition at both layers. The multi-dimensional state space and action set in our framework distinguish it from Whittle's index, motivating the development of a hierarchical indexing structure to capture complex task-server interactions.

Deriving the optimal server selection at layer $l$ via Bellman Equation \eqref{BE} is computationally challenging. However, we can conclude that due to the multi-layer nature of the MDP, each sub-problem \eqref{originrelax} exhibits a key structural property in its optimal solution ------ \textit{Multi-Layer-Threshold Type} (MLTT) structure. 
We define the critical AoI thresholds  as minimum current information age at which user $n$ selects server $m$ over $m'$ during the task execution with an age of $D$, denote as $H_n(m,m',D)$.

Denote $\Tilde{m}=\mathop{\arg\max}_{m\in\mathcal{M}} p_m$ as the index of the tentatively optimal server. Denote $p_{\pi_n(s(t))}$ as the task completion probability of the server selected by the scheduling policy $\pi_n$, which directly reflects the priority of task processing. Then, the definition of MLTT is given:

\begin{definition}[Multi-Layer-Threshold Type] If the following two conditions hold, then the offloading policy $\pi_n$ of the sub-problem \eqref{originrelax} has the Multi-Layer-Threshold Type (MLTT) structure:
\begin{enumerate}
    \item If user $n$ is at Layer 1 and the state is $s_n(t)=(\Delta_n(t)=A)$:
\begin{itemize}
    \item for any server $m'\neq \Tilde{m},m'\in\mathcal{M}\cup\{0\}$, there exists $H_n(m',\Tilde{m},0)$, $\mathrm{s.t.}$ when  $A\ge\max_{m'}H_n(m',\Tilde{m},0)$, $\pi_n(s_n(t))= \Tilde{m}$. 

    \item for any two states $s_n(t)\text{ and } s'_n(t)$ satisfying $\Delta_n(t)<\Delta'_n(t)$, there is $p_{\pi_n(s(t))}\le p_{\pi_n(s'(t))}$;
\end{itemize}

\item If user $n$ is at Layer 2 and the state is $s_n(t)=(\Delta_n(t)=A,D_n(t)=D)$:
\begin{itemize}
    \item given $D$, there exists $H_n(m,m',D)$, 
 $\mathrm{s.t.}$ when $A\ge H_n(m,m',D)$, $\pi_n(s_n(t))= m'$;
    \item for any two states $s_n(t) \text{ and } s'_n(t)$ satisfying $D_n(t)=D'_n(t) \text{ and } \Delta_n(t)<\Delta'_n(t)$, there is $p_{\pi_n(s(t))}\le p_{\pi_n(s'(t))}$;
\end{itemize}
\end{enumerate}

\end{definition}
The MLTT structure shows some common properties for optimal thresholds at both layer $1$ and layer $2$. If the user is at layer $1$, the MLTT property is equal to the Multi-Threshold Type (MTT) property in \cite{zou2021minimizing}. Since there are no ongoing tasks in Layer 1, the decision only depends on the current age $A$, without the need to consider the age $D$ at the time of task generation. Therefore, the threshold parameter is $0$. At layer $1$, the user is in a task-free state and needs to select the optimal one from all possible servers. Therefore, the thresholds of all non-optimal servers need to be covered. For Layer 2, a similar conclusion holds when the generation age \(D_n(t)\) is fixed. At Layer 2, the user is in the task execution state, and the decision-making is centered around the status of the current task. Only the comparison of specific server pairs is required, and there is no need for a global maximum. The proposition proposes that a threshold exists for a user when choosing between two servers, and the threshold depends only on the current age of the user given the same age at generation. Under preemptive scheduling, the MLTT structure takes on a distinct role. Here, the thresholds not only guide server selection but also determine task interruption points to minimize AoI. This decision mechanism leverages the unique flexibility of preemptive scheduling. 

\subsubsection{Bound of Cost-to-go Function Under Preemptive Scheduling}
To characterize the properties of the cost-to-go function, we introduce the following lemma. For simplicity, the order of $p_m$ is specified as $p_{m-1}\le p_m,\forall 1<m\le M$. 
Under the preemptive manner, we have:
\begin{lemma}[Bound of Preemptive Cost-to-go Function]\label{upperbound}
    Let $\bm{\nu}$ and $\bm{\nu}'$ be two activating cost vectors that differ only in entry by $\Delta$, i.e. state $s=(A,D)$ at layer 2. Denote $\bm{\nu}'=[\nu_1,\cdots,\nu_m+\Delta,\cdots,\nu_M],\ \forall \Delta \ge 0$. The difference between two cost-to-go functions given $\bm{\nu}$ and $\bm{\nu}'$ can be upper-bounded by
    \begin{equation}
        V_n(s,\bm{\nu}')-V_n(s,\bm{\nu})<\frac{\Delta}{p_m^2},\forall 1\le m\le M.
    \end{equation}
The difference between two cost-to-go function can be lower-bounded by
    \begin{equation}
    \begin{aligned}\label{pM+1}
         &V_n(s,\bm{\nu}')-V_n(s,\bm{\nu})>-\frac{\Delta}{p_{m+1}^2},
         \\& \forall 1\le m\le M, A\ge H_n(m,m+1,D).
    \end{aligned}
    \end{equation}
\end{lemma}
Actually, we denote $p_{M+1}=1$ here to make mathematical formulations consistent and simplify the boundary conditions. It enables the inequality in Eq. (\ref{pM+1}) to be handled without the need to separately deal with the case of $m=M$. When $m=M$, $p_{M+1}=1$ makes the lower bound expression $-\frac{\Delta}{p_{m+1}^2}$ take the form of $-\Delta$, which ensures that the lower bound is still valid and mathematically consistent even for the last server. Lemma 2 reveals the cost-to-go function bounds in the preemptive case. Following this, Proposition \ref{mltt} delves into the optimal solution of sub-problem (\ref{originrelax}). The bounds from Lemma \ref{upperbound} are closely linked to determining this optimal solution, as they shape task offloading decisions.

\begin{proposition}\label{mltt}
The optimal solution $\pi_n^*$ to the sub-problem \eqref{originrelax} is MLTT.
\end{proposition}

The proof of proposition \ref{mltt} is shown in Appendix \ref{ap1}. Fig. \ref{fig_mdpcost} illustrates the relationship between the \textit{nested index} and the MDP. Fig. \ref{fig_first_case} compares the decision of offloading a task to server $m$ at neighbor time slots for user $n$ at layer 1. In the figure, $d$ represents the specific numerical value of AoI of user $n$.
The optimal threshold $H_n(m,m',0)$ decreases as the server cost $\nu_m$ increases from 0 to $\infty$. 
For a given state \(s_{n}(t)=\Delta_{n}(t)\), there exists a minimum value for \(\nu_{m}\), which is precisely the minimum value of \(\nu_{m}\) corresponding to the case when offloading the task to server $m$ becomes the optimal decision.
Fig. \ref{fig_2_case} shows the nested index in preemptive condition, assuming that the generation age $D_n(t)$ is fixed. 
Here, the nested index quantifies the urgency of task offloading by adjusting the server cost \(\nu_m\), thereby achieving the optimality.

\begin{figure}[!t]
\centering
\subfloat[state transition Markov Chain for adopting action $y_{nm}(t)=1$ at layer $1$, with nested index defined by cost $\nu_m$ yielding threshold $H_n(m,m',0)=d+2$.]{\includegraphics[width=2.2in]{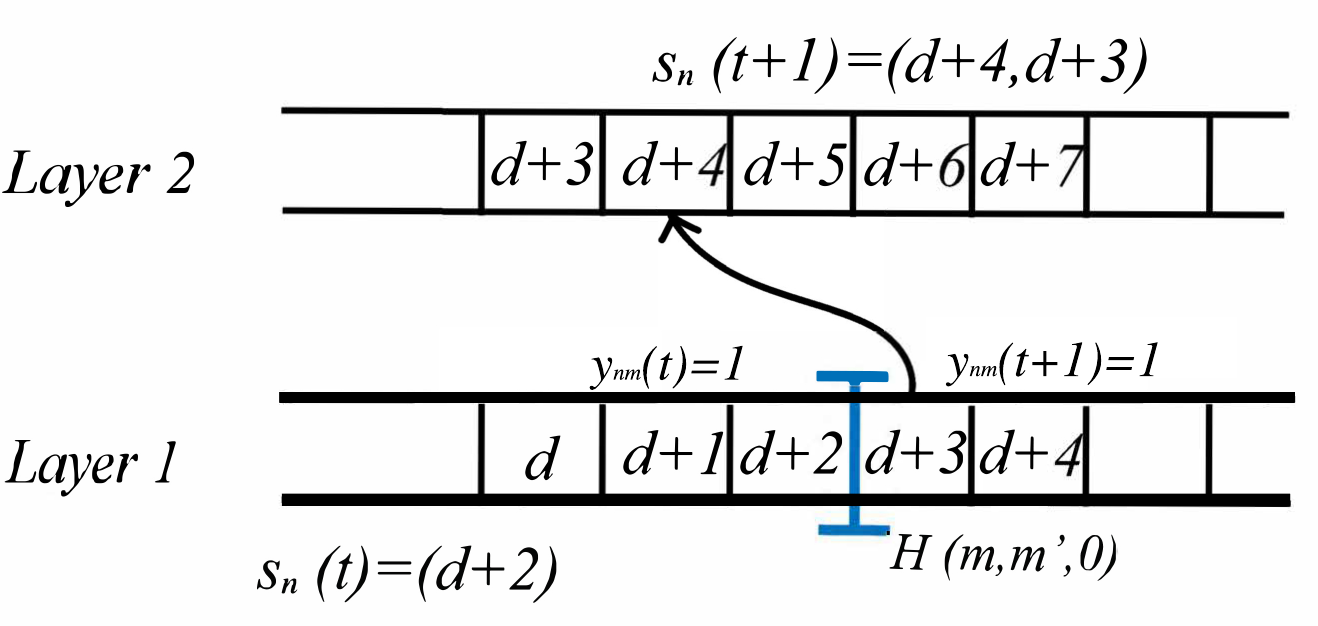}
\label{fig_first_case}}
\hfil
\subfloat[state transition Markov Chain for adopting action $y_{nm}(t)=1$ at layer $2$,  with nested index defined by cost $\nu_m$ yielding threshold $H_n(m,m',d+3)=d+6$.]{\includegraphics[width=2.4in]{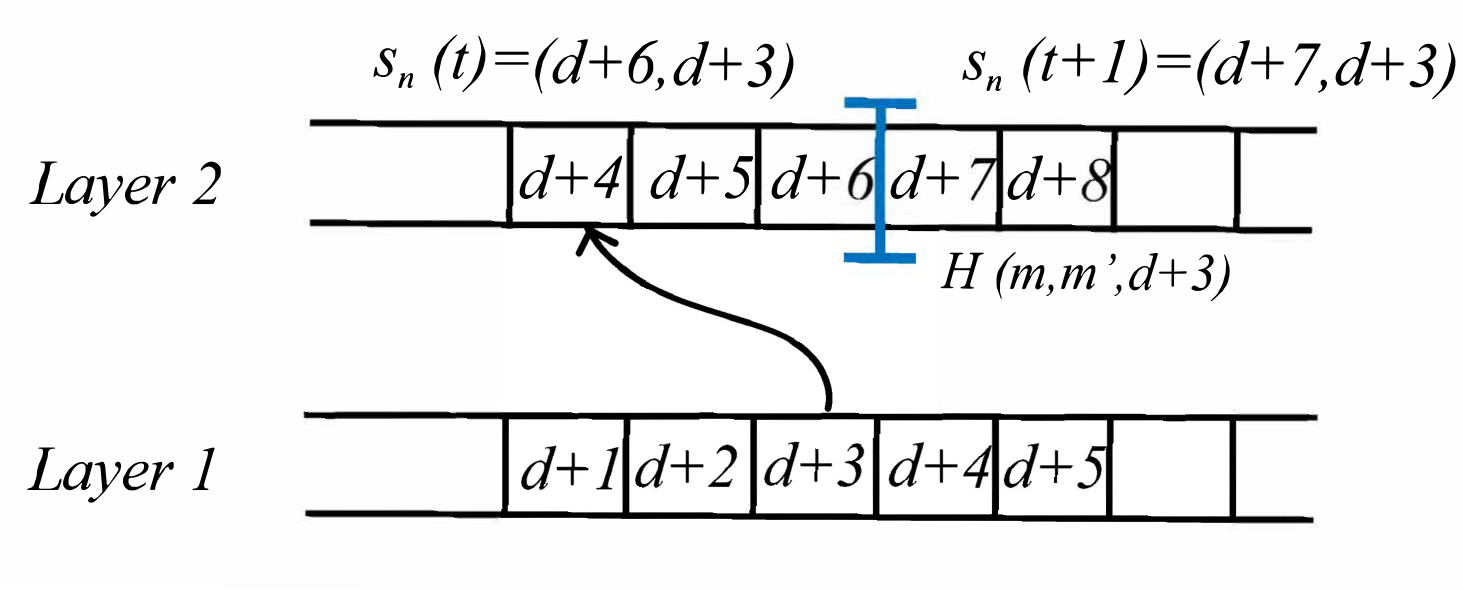}
\label{fig_2_case}}

\caption{The relationship between the \textit{nested-index} and the threshold in Markov Chain.}
\label{fig_mdpcost}
\end{figure}
\vspace{-1em}






\subsection{Nested Index Policy Under Non-Preemptive Structure}





While the previous subsection established the nested index for preemptive systems, non-preemptive task execution imposes distinct operational constraints. 
In the preemptive scheduling, the action of Layer 2 includes  ``interrupt the current task" and ``switch to another server". However, in the non-preemptive mode, the task cannot be interrupted. 

\subsubsection{The Bellman Equation and MLTT Under Non-Preemptive Structure}
The actions of Layer 2 can only be 2 types: 1) Continue to execute on the current server, denoted as $y_{nm}(t)=1$, where $m$ is the current server; 2) Abandon the current task, denoted as $y_{nm}(t)=0$. The way of transforming the action of Layer 2 from ``Selecting any server" to ``Continue on the current server" or ``Abandon the task" ensures the ``non-interruptible" constraint in the non-preemptive mode.

To address this, we reclaim the definition of actions at layer $2$ for non-preemptive conditions. 
For users at layer $1$, $y_n(t)$ represents the edge server to be chosen; for users at layer $2$, $y_{nm}(t)=1$ represents continuing execution at the server $m$ while any other action means to drop the task. The Bellman equation under the restatement is as follows:
\begin{equation}
    \gamma_n^*+V_n(s, \bm{\nu})=C_n(s,m)+\sum\limits_{s'\in\mathcal{S}_l}q^{ss'}_{nm}V_n(s',\bm{\nu}),
\end{equation}
where the summation is for the state $s'$ of the specific layer $\mathcal{S}_l$. This formulation captures the state transition dynamics and value propagation under non-preemptive constraints. Although the Bellman equations of the two scheduling mechanisms are similar in mathematical form, the non-preemptive scheduling incorporates the non-preemptive constraints (\ref{probability_group}) into the Bellman equation through $q^{ss'}_{nm}$.

For non-preemptive systems, the equilibrium condition is derived as:
\begin{equation}\label{beequal}
\begin{aligned}
&\ \ \ \ \ C_n(s,m')+\sum\limits_{s'\in\mathcal{S}_l}q^{ss'}_{nm'}V_n(s',\bm{\nu})\\&=C_n(s,m'')+\sum\limits_{s'\in\mathcal{S}_l}q^{ss'}_{nm''}V_n(s',\bm{\nu}),\ \forall m',m''\neq m.
\end{aligned}
\end{equation}
Eq. \eqref{beequal} indicates that for any two servers \(m'\) and \(m''\), their expected costs in state $s$ are the same. 
It should be noted that Eq. (\ref{beequal}) is not applicable to preemptive scheduling. In preemptive scheduling, tasks can be interrupted and reassigned to other servers, and this flexibility makes the decision-making process more complex. The state transition and decision constraints are significantly different from those in non-preemptive scheduling. The Bellman equation of preemptive scheduling needs to be adjusted in real time according to the changes in the system state, and its decision objective (Eq. (\ref{BE})) is to find the server and layer that minimize the overall cost among all possible servers $m$ and layers $l$. This is essentially different from the equality of expected costs of different servers in a non-preemptive system, as described by Eq. (\ref{beequal}).

In the MLTT property, different actions have their respective optimal thresholds related to the age of users, and these thresholds are determined by the layer $l$ and the generating age \(D_n(t)\). However, the non-preemptive condition has a more simple structure. In a generate-at-will model, it is always optimal to wait until computation finishes. By utilizing the MLTT property of the solution to problems in MEC systems, we can show the intra-indexability of the sub-problem \eqref{originrelax}.

\subsubsection{Bound of Cost-to-go Function Under Non-preemptive Scheduling}
In non-preemptive scheduling, once a task starts execution, it cannot be switched to another server. Therefore, the immediate cost before completion remains the same. However, when the server cost changes, the user at layer $1$ may choose a different server from the very start, which causes a maximum gap of $\frac{\Delta}{p_m}$. Therefore, under the condition in lemma \ref{upperbound}, we have:
\begin{lemma}[Bound of Non-preemptive Cost-to-go Function]\label{lemma3}
     The difference between two cost-to-go function under non-preemptive scheduling can be lower-bounded by
    \begin{equation}
         V_n(s,\bm{\nu}')-V_n(s,\bm{\nu})<\frac{\Delta}{p_{m}},\forall 1\le m\le M.
    \end{equation}
\end{lemma}
The boundary constraints of Lemma \ref{lemma3} ensure the stability of the threshold in the non-preemptive mode. By limiting the variation range of the cost function, the existence of the MLTT structure in the non-preemptive mode is guaranteed, that is, there exists a unique threshold $H_n(m,0,D)$. When the AoI exceeds this threshold, the expected cost of continuing to execute the current server is lower than that of abandoning the task. This structure corresponds to the structure of the preemptive scheduling MLTT, and they jointly support the nested index to achieve near-optimal decisions through the threshold strategy in both types of scenarios, enabling the nested index structure to uniformly handle the heterogeneous constraints of preemptive and non-preemptive scenarios.
The difference between \ref{upperbound} and Lemma \ref{lemma3} reflects the uniqueness of the non-preemptive mode. And these lemmas state the relationship between the server cost and the optimal state to offload a task to the corresponding server. The indexability characterizes the positive correlation between the server cost and the size of the passive set. Therefore, we can derive the intra-indexability by contradiction.
\begin{theorem}\label{intraindex}
The MDP sub-problem \eqref{originrelax} is intra-indexable given cost $\bm{\nu}$.
\end{theorem}
The proof is shown in Appendix \ref{ap2}. By analyzing the boundary properties of the cost-to-go function and the monotonicity of the passive set, the intra-indexability of the multi-layer MDP model is proved. 
Since the AoI minimizing problem in MEC systems is a $2$-layer MDP, we design thresholds for actions at both layers. The intra-indexability property ensures that each server has a unique threshold at each layer, and as the server cost increases, the size of the passive set for each server at each layer also grows. This allows the index to represent the urgency of a state and enables comparison of states across layers. Based on these properties, we define the index for the 2-layer MDP.

\subsubsection{Nested Index Under Non-Preemptive Structure}
To characterize the decision priorities in non-preemptive task scheduling, where tasks cannot be interrupted once initiated, a refined index framework is essential. Based on the multi-layer MDP model, the nested index concept is extended to accommodate the unique constraints of non-preemptive execution.
\begin{definition}[Nested Index]
\textit{The nested index for taking $y_{n}(t)=m$ at state $s_n(t)$ is defined as}
\begin{subequations}
\begin{align}\label{nestedindex}
 &I_{nm}(s_n(t), \bm{\nu})\triangleq \max\Big[0, \inf\Big\{\nu_m\ | \nonumber\\&\ \min_{m'\in\mathcal{M}}\mu_{nm'}(s_n(t),\bm{\nu}) <\mu_{nm}(s_n(t),\bm{\nu})\Big\}\Big].
 \end{align}
\end{subequations}
\end{definition}
\noindent where \(\mu_{nm}(s_{n}(t),\bm{\nu})\) represents a function related to the performance measure when taking action \(y_{n}(t) = m\) in state \(s_{n}(t)\) with server cost vector \(\bm{\nu}\). The nested index \(I_{nm}(s_{n}(t),\bm{\nu})\) is used to determine the relative urgency of user $n$ choosing server \(m\) in state \(s_{n}(t)\). Compared with the partial index \cite{zou2021minimizing}, the nested index not only considers the cost of server $m$ but also involves the costs of other servers, which can reflect the urgency of transitioning to different layers. 

In terms of the non-preemptive condition, the actions at Layer $2$ can be divided into 2 groups: keeping computation and dropping the current task. Since the optimal solution is to keep computing until the task is completed, the non-preemptive condition can be regarded as a MLTT MDP. $\forall n\in\mathcal{N},m\in\mathcal{M}$ the threshold can be expressed as $H_n(m,m',D)=1$.
According to the \textit{nested index} derived at each time slot, the central actuator can schedule the tasks of all users following the nested index policy. Define $y_{nm}$ as the decision variable for user $n$ on server $m$, which is used to indicate whether user \(n\) offloads the task to server \(m\).
At each time slot \(t \in T\), the following binary decision scheduling problem needs to be solved:
\begin{subequations}\label{eq:algo}
    \begin{align}
         \label{algoobject1}\max\limits_{\bm{u}} &\sum\limits_{n\in\mathcal{N}}\sum\limits_{m\in\mathcal{M}} I_{nm}(s_n(t),\bm{\nu}) y_{nm}&\\ \label{algocons1}
         \mathrm{s.t.}&\sum\limits_{n\in\mathcal{N}}y_{nm}\le 1,&&\forall m\in\mathcal{M}, \\\label{algocons2}
         &\sum\limits_{m\in\mathcal{M}}y_{nm}\le 1,&&\forall n\in\mathcal{N},\\\label{algocons3}
         & y_{nm}\in \{0,1\} ,\forall n\in\mathcal{N},&&\forall m\in\mathcal{M}.
    \end{align}
\end{subequations}

By maximizing the sum of the products of the nested index values of all user-server combinations and the decision variables, the optimal task offloading decision can be determined. Relax the constraint (\ref{algocons1}) into the objective function to obtain the Lagrangian dual problem:
\begin{subequations}
\begin{align}\label{lagrange_nu}
\mathcal{L}(u, \bm\nu') = &\sum\limits_{n\in\mathcal{N}}\sum\limits_{m\in\mathcal{M}}  I_{nm}(s_n(t),\bm{\nu})y_{nm}\nonumber\\& + \sum\limits_{m\in\mathcal{M}} \nu'_{m,t} \left(1 - \sum\limits_{n\in\mathcal{N}} y_{nm}\right),
    \end{align}
\end{subequations}
where $\nu'_{m,t}>0$  is the dual variable corresponding to constraint (\ref{algocons1}), representing the cost of resource usage for server $m$ at time $t$.
 The mapping process from the current state $s_n(t)$ to the action variable $y_{nm}(t)$ is named the \textit{nested index policy}.


In Algorithm \ref{policy}, the nested index $I_{nm}(s_n(t),\bm{\nu}_{t-1})$ for each user is computed  via Eq. (\ref{nestedindex}). Then, the optimal solution $y_{nm}(t)$ for problem (\ref{eq:algo}) is obtained, which is an integer  linear programming problem. Then, the solution for problem (\ref{eq:algo}) is mapped to offloading decisions and computing decisions to schedule tasks. State updates from edge servers, offloading decisions, and computing decisions are obtained. 
Subsequently, the activating cost $\bm{\nu}_{t}$ is updated. The operations from Line 4 to Line 8 are executed at each time slot until the nested policy converges. In the non-preemptive mode, the process can be further simplified. Specifically, for a user at Layer 2, according to Eq. (\ref{beequal}), there is no need to calculate the index value for all the actions. Instead, only the index values for the current server and dropping the current task need to be calculated. This is because Eq. (\ref{beequal}) indicates that in non-preemptive scenarios, the optimal action at Layer 2 is either to keep executing at the current server or to drop the task, simplifying the index-value calculation. The calculation of the index value can be very complex, and we gave an approximation of the nested index given $s_n$, $\bm{\nu}_{-m}$, $\forall\  0<m<M$, where  $\bm{\nu}_{-m}$ is the other parts of $\bm{\nu}$ except for the $m$-th item.

\begin{algorithm}[t]

\caption{A Nested Index Policy}
    \begin{algorithmic}[1]\label{policy}
        \STATE Initialize parameters $N$, $M$, $L$, $\beta$;
        \STATE Initialize $s_n(0)$ for each user $n$ and server cost $\bm{\nu}_0$;
        \FOR{$t\le T$}
            \STATE Compute nested index $I_{nm}(s_n(t),\bm{\nu}_{t-1})$ for each user $n$ and layer $l$ by Eq. \eqref{nestedindex};
            \STATE Solve the maximization problem (\ref{eq:algo}) and obtains $y_{nm}$;
            \STATE Schedule tasks according to $y_{nm}$;
            \STATE Get state updates from edge servers and actions;
            \STATE Update cost $\bm{\nu}_{t}\gets(1-\beta)\bm{\nu}_{t-1}+\beta\bm{\nu}_{t-1}'$, where $\bm{\nu}_{t}$ is the server cost at time $t$, and $\bm{\nu}'_t =\{\nu'_{1,t}, \nu'_{2,t}, \ldots, \nu'_{M,t}\}$ is a cost vector composed of the optimal dual variables of all servers from Lagrangian dual problem \eqref{lagrange_nu} at time $t$.
        \ENDFOR 
    \end{algorithmic}
\end{algorithm}

\begin{proposition}\label{indexfunc}
Given $s_n(t)=(\Delta_n(t),D_n(t))$, the index function satisfies
\begin{equation}
I_{nm}(s_n(t),\bm{\nu})=\nu_{m-1}+\Delta_n(t)-\gamma_n^*.
\end{equation}
The index for server $m$ can be derived by solving $I_{nm}(s_n(t),\bm{\nu})=\nu_m$. That is, 
\begin{itemize}
    \item if \(\nu_m < \nu_{m-1} + \Delta_n(t) - \gamma_n^*\): The cost is too high, and it is better to choose other servers;
    \item if \(\nu_m = \nu_{m-1} + \Delta_n(t) - \gamma_n^*\): It is in a critical state, and the expected cost of choosing $m$ is equal to that of the suboptimal server;
    \item if \(\nu_m > \nu_{m-1} + \Delta_n(t) - \gamma_n^*\): The cost is too high, and it is better to choose other servers.
\end{itemize}
\end{proposition}

When the value of the index function \(I_{nm}(s_{n},\bm{\nu})\) equals to the cost \(\nu_{m}\) of server \(m\), the determined index value can balance the user state and the server cost. The proof is given in Appendix \ref{ap3}. The indices for other layers can be similarly derived within finite steps of computation. 
In the non-preemptive mode, we directly determine the final server selection by designing nested index which associates the task generation time $D_n(t)$ and the server cost $\nu_m$, solving the optimization problem of the ``one-time decision" in the non-preemptive mode.
We derive the optimal average cost $\gamma_n^*$ by the technique similar to that used in \cite{tripathi_whittle_2019}, which involves solving a set of a finite number of equations. This reduces the complexity of calculating the index function and enhances the feasibility of our algorithm.

\subsection{Fluid Limit Model}
The optimality of the index policy in Algorithm 1 is established using a fluid limit argument, following the approach in \cite{verloop_asymptotically_2016}.
The fixed-point solution is the solution for the fluid limit model of the original problem. 
We will show that the fixed-point of the fluid limit model for problem (\ref{eq:algo}) is equivalent to that of problem (\ref{originrelax}).

The fluid fixed-point and fluid limit model are defined as follows. Let $z_{ns}\in [0,1]$ denote the fraction of user $n$ in state $s$, where $\sum_{s\in\mathcal{S}}z_{ns}=1$. Here, \(\mathcal{S}\) represents the set of all possible states for users. Let  $y_{nm}^s\in[0,1]$ denote the fraction of user $n$ combined with state $s$ at server $m$ given by the optimal solution of the relaxed problem \eqref{originrelax}, which is a continuous variable under the fluid model.
Let $(\bm{y}^*,\bm{z}^*,\bm{\nu}_m)$ represent the \textit{fluid fixed-point} of the following \textit{fluid limit reformulation} of problem (\ref{originrelax}):

\begin{subequations}\label{fluidx}
    \begin{align}
        \mathop{\mathrm{min}}\limits_{\bm{y},\bm{z}}& \sum_{n\in \mathcal{N}}{\sum_{s\in\mathcal{S}}}\sum\limits_{m\in\mathcal{M}}z_{ns}(\Delta_n(t)+\nu_m)y_{nm}^{s} \tag{\theparentequation a}\label{fluidx:a} \\
        \mathrm{s.t.} &\sum_{n\in \mathcal{N}}\sum_{s\in \mathcal{S}}z_{ns}y_{nm}^{s}\le 1, \tag{\theparentequation b}\label{fluidx:b}  \forall m\in\mathcal{M}, \\
        &\sum\limits_{m\in\mathcal{M}}y_{nm}^{s}\le 1, \tag{\theparentequation c}\label{fluidx:c} \forall n\in\mathcal{N}, \forall s\in\mathcal{S}, \\
        &\sum\limits_{s\in\mathcal{S}}z_{ns}= 1, \tag{\theparentequation d}\label{fluidx:d} \forall n\in\mathcal{N}, \\
        &z_{ns},y_{nm}^s\in[0,1], \tag{\theparentequation e}\label{fluidx:e} \forall n\in\mathcal{N},\forall m\in\mathcal{M},\forall s\in\mathcal{S}, \\
        &\sum_{s'\in \mathcal{S}}z_{ns}\sum\limits_{m\in\mathcal{M}}y_{nm}^{s} q^{ss'}_{nm}=\nonumber\\&\sum_{s'\in \mathcal{S}}z_{ns'}\sum\limits_{m\in\mathcal{M}}y_{nm}^{s'}q^{s's}_{nm} \tag{\theparentequation f}\label{fluidx:f}, \forall n\in\mathcal{N},\forall s\in\mathcal{S}.
    \end{align}
\end{subequations}
where \eqref{fluidx:f} is a fluid balance constraint \cite{verloop_asymptotically_2016}, which is essential for maintaining the equilibrium of the fluid limit model in terms of the flow of users among different states and servers.

Similarly, the fluid limit reformulation problem (\ref{fluidy}) for the scheduling problem (\ref{eq:algo}) can be derived. Let $(\bm{y}^*,\bm{z}^*,\bm{\nu}'_m)$ be the fixed-point solution. We have:
\begin{subequations}
\begin{align}
&\mathop{\rm{max}}\limits_{\bm{y},\bm{z} }\sum_{n\in \mathcal{N}}{\sum_{s\in\mathcal{S} }}\sum_{m\in\mathcal{M}}z_{ns}I_{nm}(s_n(t),\bm{\nu})y_{nm}^{s}\\
\label{fluidc1}
{\rm s.t.} &\sum_{n\in \mathcal{N}}\sum_{s\in \mathcal{S}}z_{ns}y_{nm}^{s}\le 1, \ \ \ \ \forall m\in\mathcal{M},\\
&\sum_{m\in\mathcal{M}}y_{nm}^{s}\le 1,\ \ \ \ \ \ \ \ \ \ \ \  \forall n\in\mathcal{N},\forall s\in\mathcal{S}.
\end{align}\label{fluidy}
\end{subequations}
The above mentioned fluid-limit reformulation of problems \eqref{fluidx} and \eqref{fluidy} are crucial for analyzing the performance of the nested index policy. Then, we can evaluate the performance of our \textit{nested index} policy based on the fluid limit model for both problems.

\begin{proposition}\label{fixedpoint}
The fixed-point solution to the problem (\ref{fluidx}) is equivalent to the solution (the fluid fixed-point) to the problem (\ref{fluidy}), i.e., we have
\begin{equation}
    (\bm{y}^*,\bm{z}^*,\bm{\nu}'_m)=(\bm{y}^*,\bm{z}^*,\bm{\nu}_m).
\end{equation}
\end{proposition}
The proof is shown in Appendix \ref{ap4}. The equivalence of fixed-point solution builds the connection between problem (\ref{eq:algo}) and problem (\ref{decoupled_problem}).
Within this proof, we first prove the sub-problem \eqref{originrelax} satisfies the \textit{precise division} \cite{partialoffloading}  property under both preemptive and non-preemptive conditions. Though problem (\ref{eq:algo}) follows an instantaneous constraint (\ref{algocons1}), its fixed-point solution still reaches the optimality at the fluid limit, which contributes to the \textit{asymptotic optimality} of our policy in Algorithm \ref{policy}. By scaling a system by $r$, we scale the number of users $N^r$ and servers $M^r$ by $r$ proportionally, i.e., let $N^r=r\cdot N$, $M^r=r\cdot M$ while keeping $\frac{N^r}{M^r}$ a constant\footnote{The system parameters $\bm{\tau}$ and $\bm{p}$ are also scaled proportionally, i.e., $\bm{\tau}^r=[\bm{\tau},\bm{\tau},\dots,\bm{\tau}]\in\mathbb{Z}_+^{1\times N^r}$ and $\bm{p}^r=[\bm{p}',\bm{p}',\dots,\bm{p}']\in[0,1]^{M^r\times N^r}$, where $\bm{p}'=[\bm{p}^T,\bm{p}^T,\dots,\bm{p}^T]^T\in[0,1]^{M^r\times N}$}.  Building on the equivalence of the fixed-point solutions, we further explore the performance of the nested index policy. Lemma 4 presents the following result:
\begin{lemma}
Under a mild global attractor assumption,
the expected objective $V^{\pi}_r$ for problem (\ref{eq:algo})  under the nested index policy $\pi$ achieves the optimal objective $V^*$ for the fluid limit model of problem (\ref{decoupled_problem}) asymptotically, i.e.,
\begin{equation}
\lim\limits_{r\to+\infty} V^{\pi}_r=V^*.
    \end{equation}
\end{lemma}
Details of the global attractor assumption can be found in \cite{verloop_asymptotically_2016}. 
It can be summarized that the system will converge to a unique steady state in the long-term operation, enabling the average age can be calculated. This assumption ensures the stability of the queue state transition. The equivalence of the fixed-point solution shows the accordance of the fluid limit model for both problems. Under a mild global attractor assumption, the objective of problem \eqref{reformulatedproblem} under our \textit{nested index} converges to the optimal cost of problem \eqref{fluidx}.

\section{Simulations}
To demonstrate the practical application of the nested index policy proposed in this paper, we consider a simplified MEC system with two heterogeneous edge servers and a single user, illustrating how the policy leverages age thresholds and server costs to minimize AoI under preemptive and non-preemptive scheduling, in line with the multi-layer MDP framework and MLTT structure. 

We assume that Server 1 is a high-performance unit with a task completion probability \(p_1=0.8\) and incurs a cost of \(\nu_1 = 5\) per time slot. In contrast, Server 2 is a low-performance  server, having a task completion probability \(p_2 = 0.5\) and incurs a cost of \(\nu_1 = 3\) per time slot. This system caters to a single user whose minimum computation time \(\tau_n^{\text{min}} = 1\), indicating that tasks can commence completion after 1 time slot.
The state space of the system is partitioned into two layers. In Layer 1, when there is no ongoing task, the state is denoted as \(s = A\), where \(A\) represents the current AoI. In Layer 2, during task execution, the state is \(s=(A, D)\), with \(A\) being the current AoI and \(D\) being the generation age, which is the AoI at the moment the task was offloaded.

\subsubsection{Case Study for Non-preemptive Scheduling}

In the non-preemptive scheduling, when in Layer 1, the user has three options for task offloading: Server 1, Server 2 or waiting. We use a virtual server 0 with \(\nu_0 = 0\) to present waiting state.
Based on Proposition 2, the nested index for server \(m\) is computed using the formula \(I_{n m}(A,\nu)=\nu_{m-1}+A - \gamma_n^*\), where \(\nu_{m-1}\) is the cost of the sub-optimal server (for the first server selection, \(\nu_{m-1} = 0\)), and \(\gamma_n^*\) is the optimal average cost per stage. 
For Server 1 with \(\nu_{m-1} = 0\), the index is \(I_{n1}=A-\gamma_n^*\), while for Server 2 with \(\nu_{m-1}=\nu_1 = 5\), the index is \(I_{n2}=5 + A-\gamma_n^*\). The decision rule is to offload to server \(m\) if \(I_{n m}\geq\nu_m\). Assuming a steady-state optimal average cost \(\gamma_n^* = 10\), the thresholds are determined as \(H_1 = 15\) for Server 1 (\(A - 10\geq5\)) and \(H_2 = 8\) for Server 2 (\(5+A - 10\geq3\)). This leads to three distinct decision regions. When \(A < 8\), the user waits to avoid unnecessary server costs, as the increase in AoI due to waiting is less costly than using either server. When \(8\leq A < 15\), Server 2 is chosen to balance its lower cost of \(\nu_2 = 3\) with its moderate completion probability of \(p_2 = 0.5 \). When \(A\geq15\), Server 1 is favored to rapidly reduce stale information despite its higher cost of \(\nu_1 = 5\).

Taking an example, suppose the user starts in Layer 1 with \(A = 5\). At time \(t = 1\), since \(A = 5<H_2 = 8\), the user waits, and the AoI increases to \(A = 6\) at \(t = 2\). When \(t = 4\) and \(A = 8\), reaching the threshold \(H_2\), the user offloads the task to Server 2 and enters Layer 2 with the state \((9,8)\). During execution on Server 2, if at a certain time \(A = 15\) and the state is \((15,8)\), since \(I_{n2}=5 + 15-10 = 10\geq\nu_2 = 3\), the user continues the execution. When the task finally completes at \(t = 10\), the AoI resets to \(t - D = 10 - 8 = 2\), and the system returns to Layer 1 for the next decision cycle. 

Under non-preemptive scheduling, as \(A\) increases, server selection transitions from waiting to Server 2 and then Server 1, in accordance with the MLTT structure, where older states prioritize servers with higher completion probabilities to minimize stale information. Low \(A\) values favor waiting to cut costs, moderate \(A\) values choose Server 2 for a cost-speed equilibrium, and high \(A\) values prioritize Server 1 to reduce AoI at the cost of higher expenditure, showcasing the rationality of the thresholds. Once a task is offloaded to Server 2, it cannot be switched, and the policy decides to continue or abandon based on the nested index.

\subsubsection{Case Study for Preemptive Scheduling}

Under preemptive scheduling, the decision-making process in Layer 1 initially mirrors that of the non-preemptive case for server selection, with the nested index calculations for Server 1 and Server 2 remaining based on \(I_{n m}(A,\nu)=\nu_{m-1}+A - \gamma_n^*\). However, the ability to interrupt tasks necessitates potential adjustments to the thresholds and decision regions. With preemptive scheduling, the system can reevaluate the server choice at each time slot, and the decision rule becomes more dynamic. If the current AoI \(A\) and the nested index of another server \(I_{n k}(A,\nu)\) (\(k\neq m\)) satisfy \(I_{n k}(A,\nu)<\nu_m\) and \(I_{n k}(A,\nu)\geq\nu_k\), the task can be preempted and offloaded to the more suitable server.

Starting again with the user in Layer 1 and \(A = 5\), the initial waiting and offloading decisions are the same as in the non-preemptive case. But consider a situation where the task is offloaded to Server 2 at \(A = 8\). Suppose at time \(t = 6\) during the execution on Server 2, the current AoI \(A = 11\). The nested index for Server 1, \(I_{n1}(11,\nu)=11 - 10 = 1<\nu_2 = 3\) and \(I_{n1}(11,\nu)\geq\nu_1 = 5\). In this scenario, under preemptive scheduling, the task on Server 2 is preempted, and the task is offloaded to Server 1. The state then changes to \((12,8)\), with the new current AoI reflecting the time elapsed during the switch, while the generation age remains \(D = 8\). Given that Server 1 has a higher completion probability, the task is more likely to be completed earlier. If the task is completed at \(t = 8\), the AoI resets to \(t - D = 8 - 8 = 0\), and the system returns to Layer 1. 

Under preemptive scheduling, the decision thresholds are continuously re-evaluated, allowing the system to utilize the high-performance Server 1 earlier when necessary, without being restricted by non-preemptive limitations. Although preemptive scheduling offers more flexibility, it also incurs additional costs related to task interruption and re-offloading. Nevertheless, the nested index policy still balances these costs with the objective of minimizing AoI, ensuring efficient server resource utilization.

\section{Numerical Results}



In this section, numerical researches were conducted to evaluate the performance of the nested index algorithm and verify its convergence property. We simulate a MEC system with initial tasks of $N=50$ users that can be divided into $6$ groups, with 
$\bm{\tau}^{min}=[2, 4, 8, 16, 32, 64]$, and successful updating probability of $\bm{p}=[0.8, 0.7, 0.6, 0.5, 0.3, 0.1]$, each of which has the number of $[5,10,5,5,10,15]$ users respectively. We set $\beta=50$ and simulated $T=10000$ slots. In the simulation, we mainly test the performance of the \textit{nested index} when solving a multi-layer MDP. 
\vspace{-1em}

\begin{figure}[htbp]
    \centering
    \subfloat[Dual cost update in the non-preemptive structure
\label{cost_nonpreemptive}]{\includegraphics[width=2.8in]{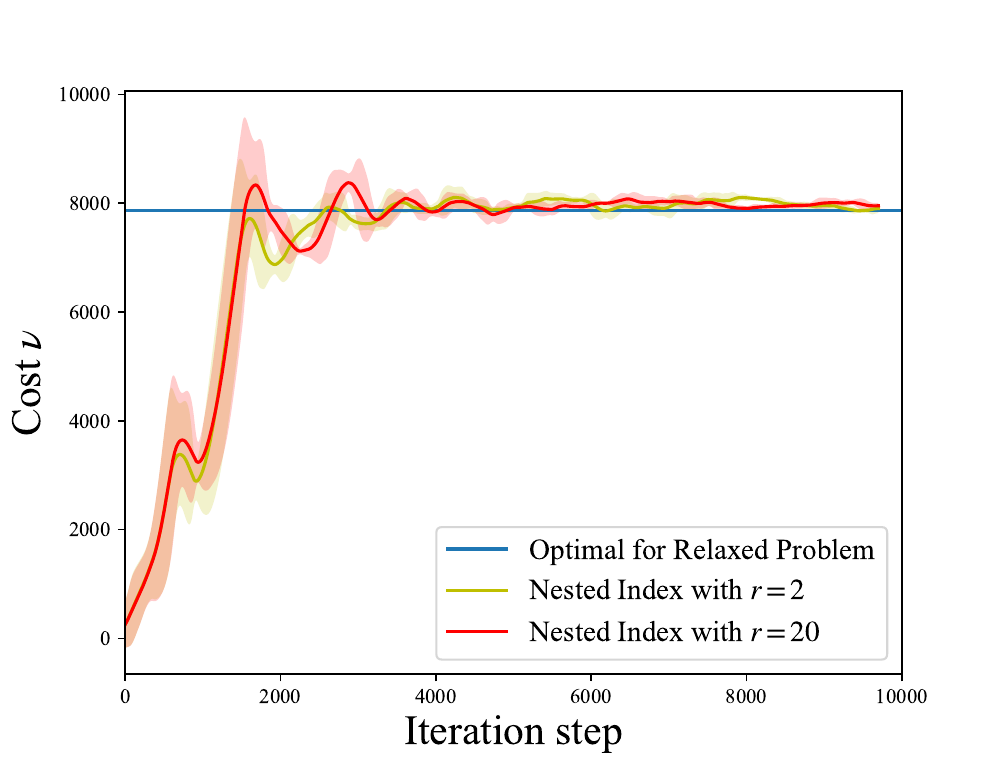}}
    \hfil
    \subfloat[Dual cost update in the preemptive structure\label{cost_preemptive}]{\includegraphics[width=2.8in]{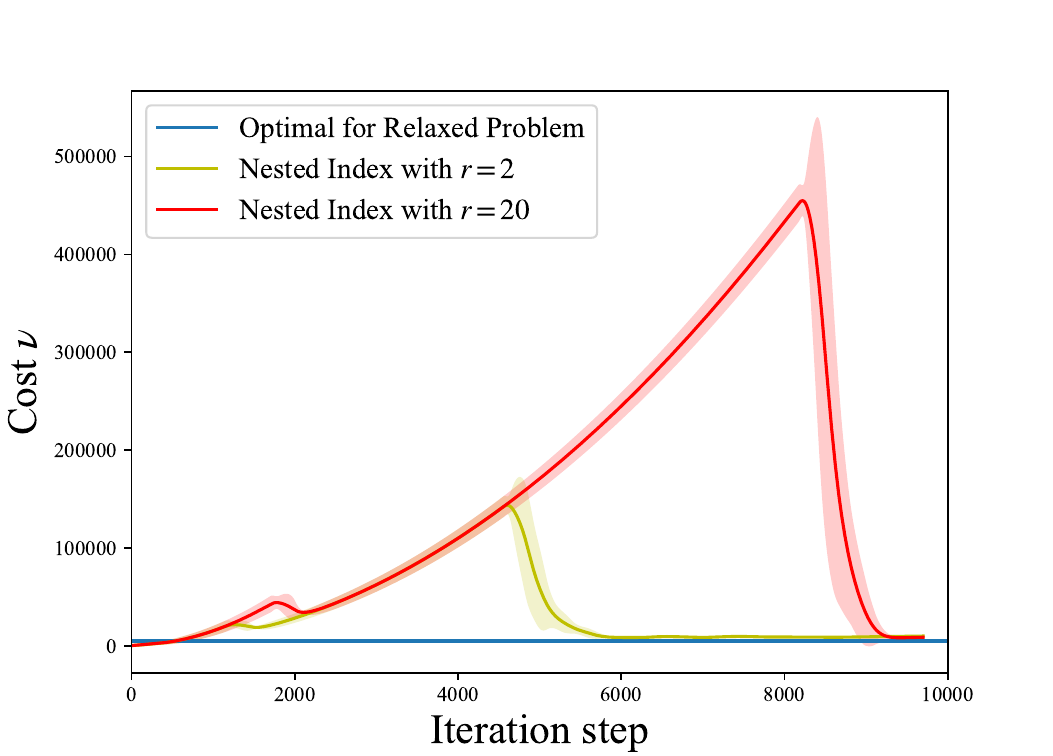}}    
    \caption{Dual cost update for the proposed index-based policy, which is compared with the dual cost update of the relaxed problem.}
    \label{converge_lambda}
\end{figure}

\subsection{Numerical Results Under Non-preemptive Scheduling}
\subsubsection{Convergence of the Cost Update}
The cost update dynamics of the proposed nested index policy are compared with those of the optimal solution to problem (\ref{originrelax}). The optimal solution in Fig. \ref{cost_nonpreemptive}
represents the dynamics of the server cost $\nu$ of the relaxed problem. A new cost \(\nu(t)\) is obtained at each time slot via dual gradient ascent. As shown in Fig. \ref{cost_nonpreemptive}, the server cost converges smoothly to a small neighborhood around the optimal cost, demonstrating the stability and effectiveness of the proposed policy in cost adjustment.
Subsequent analysis verifies the server cost dynamics of the proposed index-based policy under increasing system scales. Fig. \ref{cost_nonpreemptive} represents the dynamics of the cost update at scale $r=2$ and $r=20$. With the increase of the scalar $r$, the cost for the proposed index-based policy approaches is close to the optimal value of the dual cost.

\subsubsection{Average AoI Performance}

The average AoI performance of the proposed policy is evaluated by adopting the optimal solution to problem (\ref{originrelax}) as the lower bound of the index policy \cite{verloop_asymptotically_2016} for comparison purposes. Three benchmark policies are considered in the experimental analysis:\textit{Max-Age Matching Policy} (MAMP), \textit{Max-Age Reducing Policy} (MARP), both of which are greedy policies, and \textit{Rounded Relax Policy} (RRP). Each of these policies employs distinct task selection strategies:

\begin{itemize}
    \item The MAMP selects the user with the current maximum age and allocates a server with a higher success probability to them.

    \item The MARP takes the transition probability of MC into consideration. Define the weight in this policy as $w_{n}=\Delta_n(t)+\frac{1}{p_n}(\Delta_n(t)-D_n(t))$, which represents the probable approximation of optimality gap reduction.
    \item The RRP is derived from the solution of the relaxed problem. 
    The RRP chooses users uniformly at random to satisfy the feasibility when violating the constraint \eqref{action}.
\end{itemize}

\begin{figure}[htbp]
    \centering
    \subfloat[The average AoI in non-preemptive structure\label{converge_non}]
        {\includegraphics[width=2.75in]{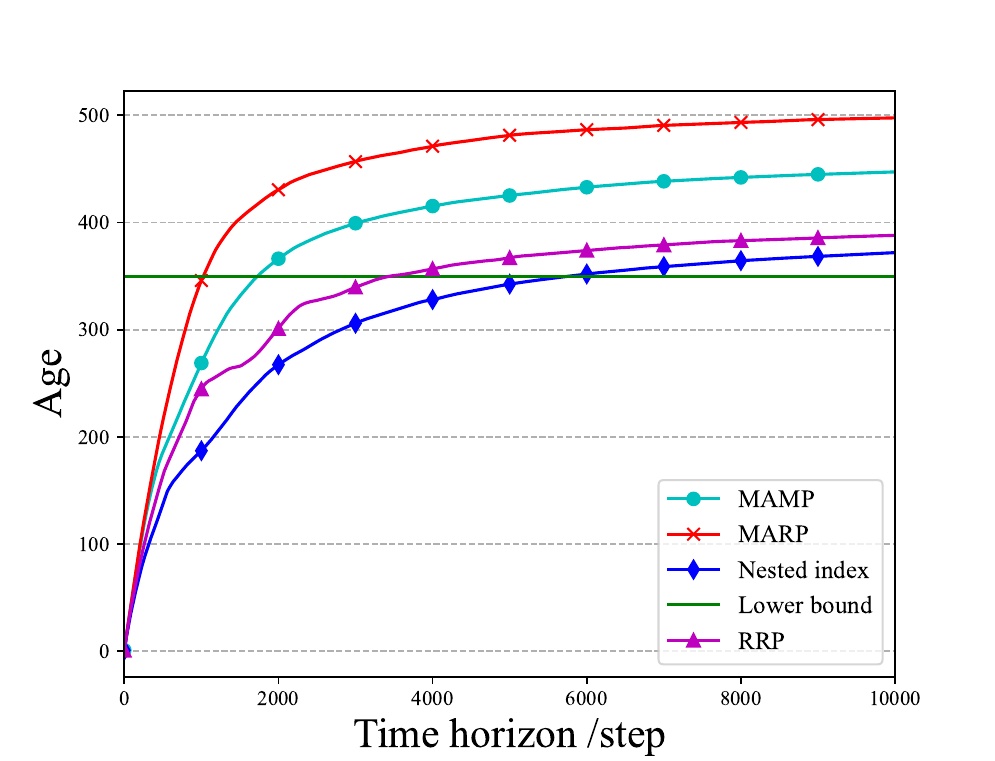}}
    \hfil
    \subfloat[The average AoI in preemptive structure\label{converge_pre}]
        {\includegraphics[width=2.8in]{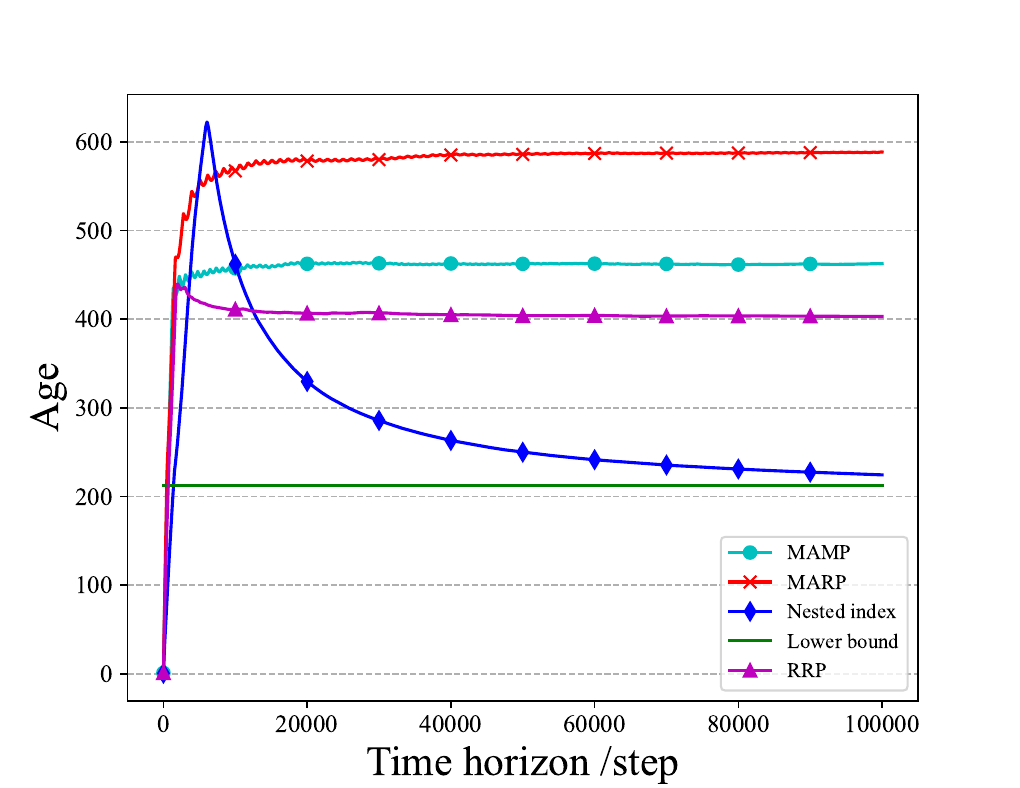}}
    
    \caption{Average AoI performance during computing.}
    \label{aoi_converge}
\end{figure}

Fig. \ref{converge_non} evaluates the average AoI under different policies such as nested index policy, MAMP, MARP, RRP, and lower bound of problem \eqref{originrelax} with $r=20$. These results clearly demonstrate the superiority of the nested index policy over the benchmark policies. 
In Table \ref{table_non}, we present the different strategy performance in non-preemptive Structure.
 The greedy policy MAMP is $17.03\%$ worse than the \textit{Nested Index} policy, MARP is $25.43\%$ worse than \textit{Nested Index} policy and RRP is $4.35\%$ worse than our approach when $r=20$, demonstrating the effectiveness of our method.
The normalized system AoI gets closer to that of the optimal AoI for the relaxed problem with the increase of the system scalar $r$. Fig. \ref{aoi_nonpreemptive} shows the normalized AoI of the system, i.e., the average age per user. The normalized AoI decreases almost monotonically as $r$ increases, further validating the effectiveness of our proposed policy in reducing the AoI as the system scale expands.

\begin{table}[htbp]
\footnotesize
  \centering
  \caption{Strategy Performance Comparison in Non-preemptive Structure \\ 
           (Last 500 steps with $ r = 20 $)}
  \label{table_non}
  \begin{tabular}{| l | r | r | r |}  
    \hline
    Strategies      & Average AoI & Absolute Gap  & Relative Gap \\
    \hline
    Nested Index    & 370.05      & --                        & --           \\
    RRP             & 386.87      & 16.82                    & 4.35\%       \\
    MAMP            & 446.03      & 75.88                    & 17.03\%      \\
    MARP            & 496.22      & 126.76                   & 25.43\%      \\
    \hline
  \end{tabular}
  \footnotetext{Compared with Nested Index}
\end{table}

\subsection{Numerical Results Under Preemptive Scheduling}
\subsubsection{Cost Convergence}
For preemptive condition, the nested index strategy demonstrates an effective ability to regulate costs during the iterative process in Fig. \ref{cost_preemptive}. When scale $r=2$ and $r=20$, the cost curve of the indexing strategy gradually rises with the increase of the number of iterative steps, and the cost gap with the Optimal solution remains stable, indicating that the strategy can approach the theoretical optimal level when allocating resources of a fixed scale. Overall, the indexing strategy achieves a gradual and stable cost through iterative optimization in the preemptive mechanism, verifying its effectiveness in the scenarios of dynamic task interruption and reallocation.

\begin{figure}[htbp]
    \centering
    \subfloat[The average AoI in the non-preemptive structure varies with the system scale
\label{aoi_nonpreemptive}]{\includegraphics[width=2.65in]{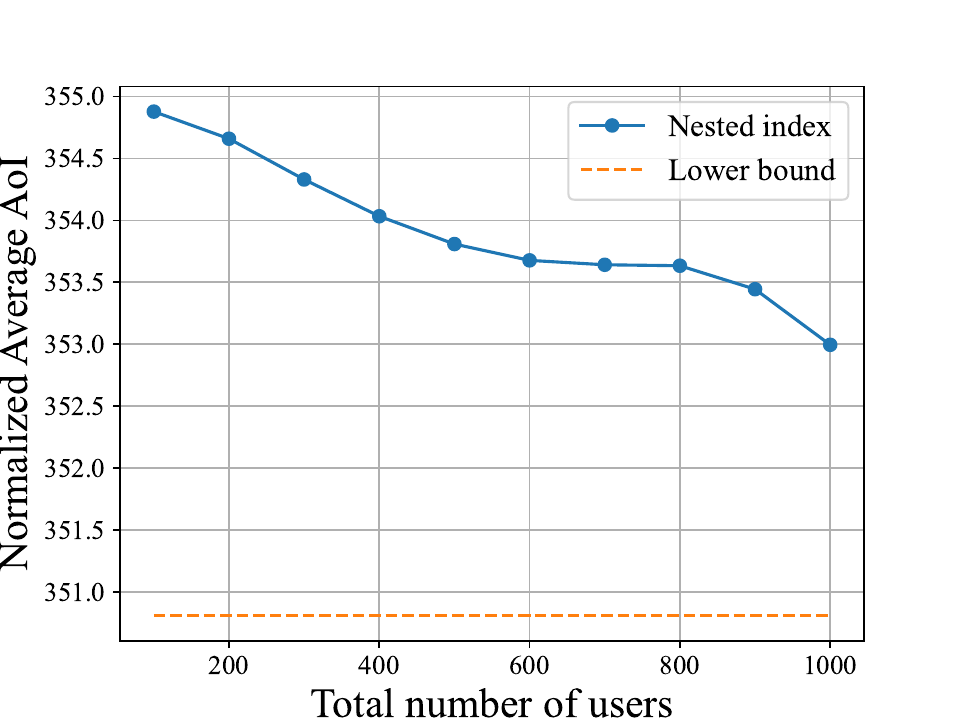}}
    \hfil
    \subfloat[The average AoI in the non-preemptive structure varies with the system scale
\label{aoi_preemptive}]{\includegraphics[width=2.65in]{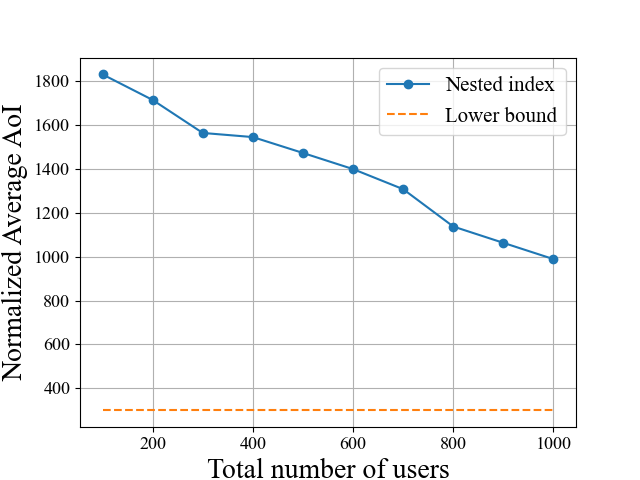}}
    \caption{Average AoI vs. system scale.}
    \label{fig_sim}
\end{figure}

\subsubsection{Average AoI Performance}
As the number of users increases, the normalized average AoI shows a convergent trend and tends to stabilize (Fig. \ref{converge_pre}). Compared with the Lower bound, the AoI of the nested indexing strategy is always higher than the theoretical optimal lower bound. This is due to the overhead of task preemption and decision constraints in the actual scheduling. However, the overall trend still remains to converge close to the lower bound, which still reflects the scalability and performance advantages of the strategy in large-scale user scenarios.
Compared with Fig. \ref{aoi_nonpreemptive}, Fig. \ref{aoi_preemptive} shows the same monotonic trend, further verifying the effectiveness of the nested indexing strategy in the preemptive structure.

In Table \ref{table_pre}, we show the performance of different strategies in the preemptive structure. Compared with other strategies, the Nested Index performs relatively poorly. When $r = 20$, the Nested Index is 48.12\% better than the greedy strategy MAMP, 88.35\% better than the MARP strategy, and 29.04\% better than the RRP strategy.

\begin{table}[htbp]
\footnotesize
  \centering
  \caption{Strategy Performance Comparison in Preemptive Structure \\ 
           (Last 500 steps with $ r = 20 $)}
  \label{table_pre}
  \begin{tabular}{| l | r | r | r |}  
    \hline
    Strategies      & Average AoI & Absolute Gap  & Relative Gap \\
    \hline
    Nested Index    & 224.55    & --                        & --           \\
    RRP              & 403.15      & 178.60      & 44.30\%            \\
        MAMP             & 462.74      & 238.20        & 61.84\%           \\
        MARP             & 588.43     & 363.88      & 51.47\%           \\
    \hline
  \end{tabular}
  \footnotetext{Compared with Nested Index}
\end{table}

\section{Conclusion}
In this study, we explored the minimization of AoI in a MEC system with heterogeneous servers and users, incorporating both preemptive and non-preemptive task scheduling mechanisms. The problem was formulated as a 2-layer MDP, explicitly modeling the distinct dynamics of these mechanisms. A novel nested index framework was introduced to account for the differing behaviors of the two scheduling mechanisms, enabling optimal decision-making under both modes. We devised scheduling polices that employs the nested index, ensuring the asymptotic optimality of the average expected AoI of the MEC system as the system scale expands. Additionally, the computation of the nested index was derived, with the advantage of lower computational complexity. Through simulation, it was proven that the proposed algorithm converges in both scheduling modes and especially provides near-optimal performance in non-preemptive structure, verifying its effectiveness in practical MEC deployments with different task requirements.

In the future, we will conduct more in-depth exploration in the aspects of model complexity and collaborative scheduling. Based on the current multi-layer MDP and nested index strategy, it is significative to consider more dimensional state variables and constraint conditions. Meanwhile, exploring the optimization problem of collaborative scheduling among multiple users and multiple servers is also an interesting direction.


\bibliographystyle{IEEEtran}
\bibliography{tmc_main}






%

\clearpage
\begin{appendices}
\section{}
{\centering\section*{Proof of Lemma \ref{detsta}}}
Recall the Lemma 1: \textit{The optimal solution to each sub-problem \eqref{originrelax} is deterministic stationary.} 
According to \cite{sennott1989average}, if we can proof the optimality of sub-problem satisfying \textit{finite average cost} and \textit{bounded relative cost function}, then the Lemma 1 holds. Finite average cost means that there exists a deterministic stationary policy for each MDP such that the average cost is finite. Bounded relative cost function means there exists a non-negative number $L$ such that the relative cost function $V(s)-V(s_0)\ge-L$.

First, we show the finite property of average cost. The average age under this policy can be calculated by:
\begin{equation}
    \begin{aligned}
      A_{ave}=p_m\cdot 1+\sum\limits_{i=2}^\infty p_m(1-p_m)^{i-1}(i+A_{ave})
\end{aligned}
\end{equation}
Solving for $A_{ave}$, we obtain $A_{ave}=\frac{1}{p_m^2}$.
Since $p_m>0,\forall\ m\in\mathcal{M}$, the average age is finite. Meanwhile, the average server cost per step is bounded by $\nu_m$ under this policy, ensuring the average cost $A_{ave}+\nu_m=\frac{1}{p_m^2}+\nu_m$ is also finite. 

Then we consider the boundary of relative cost function.
If there exists a state $s_0$ which has a minimum cost-to-go value among any state in $\mathcal{S}$, then relative cost function is bounded. Before that, we first show a lemma on the monotonicity of the cost-to-go function:  
\begin{lemma}[\textbf{Monotonicity of Cost-to-Go}]\label{monotonicincrease}
Denote the state of a user at layer $2$ as $s=(A,D)$. The cost-to-go function under $\pi^*$ is non-decreasing \text{w.r.t.} $A$:
\begin{equation}\label{monot}
\begin{aligned}
   &A\le A',D=D' \\&\Rightarrow  V((A,D),\bm{\nu})\le V((A',D),\bm{\nu})
\end{aligned}
\end{equation}
and satisfies:
\begin{equation}\label{monot2}
    \begin{aligned}
   &V((A,D),\bm{\nu})-V((A-D),\bm{\nu})
   \\&\le V((A',D),\bm{\nu})-V((A'-D),\bm{\nu}).
    \end{aligned}
\end{equation}
\end{lemma}

\begin{proof}
The proof exploits Proposition 3.1 in \cite{jiang2015approximate}, which extends the average cost within finite steps to the infinite horizon. Define $f:\mathcal{S}\times\mathcal{A}\times\mathcal{W}\to\mathcal{S}$ as the state transition function, where $s_n(t+1)= f(s_n(t),\bm{y}_n(t),w_t)$ and $w_t\in\mathcal{W}$ is the information process at time $t$. 
First, with given $s'=(A',D')$ and $s\preccurlyeq s'$, where ``$\preccurlyeq$" denotes age ordering \cite{jiang2015approximate}, we note the sub-problem satisfying the following conditions:
\begin{enumerate}
    \item{State Transition Monotonicity:} For every $\bm{a}\in\mathcal{A}$ and $w\in\mathcal{W}$, the state transition function satisfies $f(s,\bm{a},w)\preccurlyeq f(s',\bm{a},w)$ when tasks remain uncompleted. Since state $s$, the next state could be either $(A+1,D)$ or $(A+1-D)$. We have $(A+1,D)\preccurlyeq (A'+1,D')$ and $A-D\le A'-D$ at the same time given $D=D'$;
    
    \item{Stage Cost Monotonicity:} Let $g_{m}(s)=C(s,\bm{a})=\nu_m+(1-p_m)A+p_m(A-D)=\nu_m+A-p_m\cdot D$ denote as the per stage cost. Then we have \(g_m(s) \leq g_m(s')\) for \(A \leq A'\).
  
\end{enumerate}
By Proposition 3.1 in \cite{jiang2015approximate} and above conditions, for any finite $T$, we have $V_T((A,D),\bm{\nu})\le V_T((A',D),\bm{\nu})$. By utilizing the convergence of the value iteration \cite{jiang2015approximate, optimalcontrol}, we derive the inequality \eqref{monot}.

To establish inequality \eqref{monot2}, we first assume that for the \((t+1)\)-stage value function, the difference \(V_{t+1}((A, D), \nu) - V_{t+1}((A-D), \nu)\) is non-decreasing in $A$, i.e., 
\begin{equation}
\begin{aligned}
         V_{t+1}((A,D),\bm{\nu})-V_{t+1}((A-D),\bm{\nu})\\\le V_{t+1}((A',D),\bm{\nu})-V_{t+1}((A'-D),\bm{\nu})
    \end{aligned}
\end{equation}
For any actions \(a, b\), because of the monotonicity of \(V_{t+1}\), larger $A$ leads to non-decreasing future cost differences and it can be bounded by the inductive hypothesis:
\begin{equation}\label{expectedt+1}
    \small\begin{aligned}
        &\mathbb{E}\left[V_{t+1}(f((A,D),\bm{a},w_{t+1}),\bm{\nu})\mid S_t=(A,D),a_t=\bm{a}\right]
        \\&-\mathbb{E}\left[V_{t+1}(f((A-D),\bm{b},w_{t+1}),\bm{\nu})\mid S_t=(A-D),a_t=\bm{b}\right]\\
        &\le \mathbb{E}\left[V_{t+1}(f((A',D),\bm{a},w_{t+1}),\bm{\nu})\mid S_t=(A',D),a_t=\bm{a}\right]
        \\&-\mathbb{E}\left[V_{t+1}(f((A'-D),\bm{b},w_{t+1}),\bm{\nu})\mid S_t=(A'-D),a_t=\bm{b}\right]
    \end{aligned}
\end{equation}
Then we define the $T$-stage value function via the Bellman equation, which optimally selects the action \(a \in \mathcal{A}\) to minimize the sum of immediate cost \(C(s, a)\) and expected future cost:
\begin{equation}\label{stagecostt}
\begin{aligned}
    &V_t(s,\bm{\nu})=\min\limits_{\bm{a}\in\mathcal{A}}\big[ C(s,\bm{a})
    \\&+\mathbb{E}\left[V_{t+1}(f(s,\bm{a},w_{t+1}),\bm{\nu})\mid S_t=s,a_t=\bm{a}\right]\big].    
\end{aligned}
\end{equation}
And the difference in immediate costs between states \((A, D)\) and \((A-D)\) is invariant to $A$, depending only on server costs and task parameters:
\begin{equation}\label{invariance}
\begin{aligned}
    &C((A,D),\bm{a})-C((A-D),\bm{b})\\
    =&\nu_m-\nu_{m'}+(1-p_m)D\\
    =&C((A',D),\bm{a})-C((A'-D),\bm{b}),\forall \bm{a},\bm{b}\in\mathcal{A},    
\end{aligned}
\end{equation}
Substituting the Bellman equation (\ref{stagecostt}) into the expectation inequality (\ref{expectedt+1}) with the cost invariance (\ref{invariance}), then for state $(A, D)$, we have $V_t((A, D), \nu) = \min_a \left[ C((A, D), a) + \mathbb{E}[\cdot] \right]$, and for state $(A-D)$,  we have $V_t((A-D), \nu) = \min_b \left[ C((A-D), b) + \mathbb{E}[\cdot] \right]$. Subtracting these two equations and applying (\ref{invariance}), the optimal choice does not destroy the monotonicity of future cost differences:
    \begin{equation}
        \begin{aligned}
            V_t((A, D), \nu) - V_t((A-D), \nu) \\\leq V_t((A', D), \nu) - V_t((A'-D), \nu).
        \end{aligned}
    \end{equation}
It holds for all $t$ by using backward induction, i. Hence, the lemma follows.
\end{proof}

Lemma \ref{monotonicincrease} demonstrates that the cost function $V((A,D),\bm{\nu})$ is monotonic increasing w.r.t. $A$. To proof the boundary of relative cost function, we assert that the value function $V((A,D),\bm{\nu})$ is also monotonic increasing w.r.t. $\nu_m,\forall\ m\in\mathcal{M}$, and it can be proved by contradiction: If $V((A,D),\bm{\nu})$ is not increasing in $\nu_m$, then there exist a $\nu'_m>\nu_m, \text{ s.t. } V((A,D),\{\bm{\nu}_{-m},\nu'_m\})<V((A,D),\bm{\nu}))$. However, by the definition of server cost $\bm{\nu}$, if we apply the optimal policy under $\{\bm{\nu}_{-m},\nu'_m\}$ to the condition of server cost $\bm{\nu}$, we will obtain a lower cost-to-go due to $\nu'_m>\nu_m$, which contradicts the optimality. Therefore, the value function is also increasing w.r.t. $\nu_m,\forall\ m\in\mathcal{M}$. With Lemma 5 established, the boundary of relative cost function holds by choosing \(s_0\) as the minimal-age state, where \(V(s_0, \nu)\) serves as the reference point with \(L = \max_s |V(s, \nu) - V(s_0, \nu)|\), which is finite due to finite average cost. Hence, the result of Lemma 1 follows.

\section{}\label{ap1}
{\centering\section*{Proof of Proposition \ref{mltt}}}
Recall the Proposition 1: \textit{The optimal solution $\pi_n^*$ to the sub-problem \eqref{originrelax} is MLTT.} According to the definition of MLTT in Definition 4, Proposition 1  needs to be proved from the following two aspects:
\begin{itemize}
    \item The existence of critical thresholds for the decisions of Layer 1 and Layer 2;
    \item The task completion probability $p_{\pi_n(s(t))}$  changes monotonically with $A$.
\end{itemize}
We will show the proof of MLTT for above two aspects under preemptive and non-preemptive scheduling respectively. To streamline notation, we drop subscript $n$ to simplify the notation as the proposition holds for each sub-problem.

\subsection{Proof of MLTT under Preemptive Scheduling}

To prove that the optimal solution is MLTT under preemptive scheduling, we first show that the expected cost-to-go differences is strictly increasingly in $A$.  For any user with state $s=A$ at Layer 1, the expected cost-to-go differences with higher task completion probability for server $m$ (i.e., $p_m> p_{m'}$) is as follows:
\begin{equation}
\label{19}
\begin{aligned}
       &\mu_{m'}(A,\bm{\nu})- \mu_m(A,\bm{\nu})\\=&\ \ 
     \tau^{min}\cdot(\nu_{m'} -\nu_{m})\\
     &+(1-p_{m'})[A+\tau^{min}+V((A+\tau^{min}+1,A),\bm{\nu})]\\
     &-(1-p_m)[A+\tau_{n}^{min}+ V((A+\tau^{min}+1,A),\bm{\nu})]\\
     &+p_{m'} V(\tau^{min},\bm{\nu})-p_{m} V(\tau^{min},\bm{\nu}).
\end{aligned}
\end{equation}
Isolating terms involving the current age $A$, we rewrite it as:
\begin{equation}\label{simpmono1}
    \begin{aligned}
           &\mu_{m'}(A,\bm{\nu})- \mu_m(A,\bm{\nu})=
           \cdots
           \\&+(p_m-p_{m'})[A+ V((A+\tau^{min}+1,A),\bm{\nu})\\&- V(\tau^{min}+1,\bm{\nu})],
    \end{aligned}
\end{equation}
where the constant terms independent of $A$ are omitted. According to Lemma 5, given $D=A$, we can easily derive that $V((A+\tau^{min}+1,A),\bm{\nu})$ is strictly increasing w.r.t. $A$. Since $p_m>p_{m'}$, Eq. \eqref{simpmono1} is strictly increasing w.r.t. $A$. 
When $A\rightarrow 0$, there is
\begin{equation}
\begin{aligned}
       &\mu_{m'}(A,\bm{\nu})- \mu_m(A,\bm{\nu})\\\approx&\ \ 
     \tau^{min}\cdot(\nu_{m'} -\nu_{m})+\tau^{min}\cdot(p_m-p_{m'})\\
     \end{aligned}
\end{equation}
The probability $p_m \in [0, 1]$, while $\nu_m$ is usually greater than 1. Thus,  the difference of $\nu_m$ dominates the total cost, i.e., $\mu_{m'}(A,\bm{\nu})- \mu_m(A,\bm{\nu})<0$. When $A \rightarrow \infty$, according to \eqref{simpmono1}, age cost dominates and the cost difference becomes positive, i.e., $\mu_{m'}(A,\bm{\nu})- \mu_m(A,\bm{\nu})>0$. 
Since the state space is a discrete set of natural numbers, there must exist a certain smallest \(A = H_n(m, m', 0)\) such that the cost difference is non-negative for the first time, that is, the threshold exists.


Then we show the monotonicity of task completion probability at  Layer 1 under preemptive scheduling. 
For any two states \(s(t)\) and \(s'(t)\) where \(\Delta_n(t)=A_1 < A_2=\Delta'_n(t)\), if \(A_1 <  H_n(m, m', 0) < A_2\), then \(p_{\pi_n(A_1)} = p_{m'}\)(low completion probability) and \(p_{\pi_n(A_2)} = p_m\) (high completion probability), satisfying  \(p_{m'} < p_m\). If \(A_1, A_2 <  H_n(m, m', 0)\) or \(A_1, A_2 \geq  H_n(m, m', 0)\), then both choose the low completion probability server, \(p_{\pi_n(A_1)} = p_{\pi_n(A_2)} = p_{m'}\). In conclusion, for any  \(A_1 < A_2\), there must be \(p_{\pi_n(A_1)} \leq p_{\pi_n(A_2)}\),  that is, the priority increases monotonically with age.


For the existence of critical threshold with state $s=(A,D)$ at Layer 2, the expected cost difference between selecting server $m$ and $m'$ can be expressed as:
\begin{equation}\label{simpmono}
    \begin{aligned}
           &\ \ \ \ \ \mu_{m'}((A,D),\bm{\nu})- \mu_m((A,D),\bm{\nu})
           \\&=\cdots+(p_m-p_{m'})[A+ V((A+1,D),\bm{\nu})
          \\&\ \ \ \ \ -V(A-D+\tau^{min}+1,\bm{\nu})].
    \end{aligned}
\end{equation}
Here, $A-D$ represents the elapsed computation time. As $A$ increases, the elapsed computing time $A-D$ gradually approaches the minimum computing time $\tau^{min}$,  and the marginal cost of delaying the task completion decreases. That is, the impact of $A \rightarrow A+1$ on the long-term average cost gradually weakens, resulting in a monotonically decreasing on marginal cost. 
When $A \rightarrow D+\tau^{min}$, the immediate cost advantage of the low-completion-probability server dominates, and \(\mu_{m'}((A,D),\bm{\nu})- \mu_m((A,D),\bm{\nu}) > 0\) (choosing $m'$ is more preferable). When $A\rightarrow \infty$, the future cost advantage of the high completion probability server (faster reduction of age) is gradually emerging, and due to the decreasing marginal cost, $\mu_{m'}((A,D),\bm{\nu})- \mu_m((A,D),\bm{\nu})$ monotonically decreases as $A$ increases, leading to $\mu_{m'}((A,D),\bm{\nu})- \mu_m((A,D),\bm{\nu})<0$. Therefore there exists a unique critical threshold $H(m,m',D)$.


Similarly, in Layer 2, given the generation age $D$, as $A$ increases, the optimal policy transitions to servers with higher \(p_m\), maintaining \(p_{\pi_n(s(t))} \leq p_{\pi_n(s(t'))}\) for \(\Delta_n(t) < \Delta'_n(t)\). 
These properties ensure that the sub-problem \eqref{originrelax} exhibits MLTT in preemptive scheduling.

\subsection{Proof of MLTT under Non-preemptive Scheduling}
Next, we consider the proof under non-preemptive condition. The difference between preemptive and non-preemptive conditions lies in the second layer of multi-layer MDP. 
For non-preemptive condition, the $M+1$ actions also can be divided into 2 groups, i.e., continuously computing, or interrupts computing on the server $m$ and returns to Layer 1.
Since a task can not be switched to another server during computation, all other actions can be considered as dropping the task, which will cause the user's status to directly return from $(A,D)$ in Layer 2 to Layer 1, and the age is updated to $A+1$. 
Using $m' = 0$ to represent the virtual server that abandons the task. The decision at Layer 2 under non-preemptive manner can be regarded as the server's comparison of $(m,0)$, corresponding to the threshold $H_n(m,0,D)$ in Definition 4 of the main text. Specifically, for the non-preemptive sub-problem (3), it is necessary to prove the existence of a threshold $H_n(m,0,D)$. When the current age $\Delta_n(t) = A \geq H_n(m,0,D)$, the expected cost of continuing to execute the current server $m$ is lower; conversely, it is better to choose the virtual server $0$, i.e., giving up the task. 
The difference in the expected costs between the two is:
\begin{equation}
    \begin{aligned}
        &\mu_{nm}(A, D, \nu) - \mu_{n0}(A, D, \nu) \\=& \left[\nu_m + (1 - p_m)V_n((A+1, D), \nu)\right] - \left[A + 1 + V_n(A+1, \nu)\right]
        \\=& (\nu_m - A - 1) + (1 - p_m)\left[V_n((A+1, D), \nu) - V_n(A+1, \nu)\right]
    \end{aligned}
\end{equation}

By Lemma 5, both \(V_n((A+1, D), \nu)\) and  \(V_n(A+1, \nu)\) are non-decreasing with respect to A, and \(V_n((A+1, D), \nu) \geq V_n(A+1, \nu)\). 
And the immediate cost difference  \(\nu_m - A - 1\) strictly decreases with $A$ leading the overall cost difference to monotonically decrease. Therefore, \(\mu_{nm}(A, D, \nu) - \mu_{n0}(A, D, \nu) \) strictly decreases as $A$ increases. When $A \rightarrow D$, the immediate cost of giving up the task $A+1$ is low, while the high cost of continuing to execute $\nu_m$ dominates, and the cost difference is positive (giving up is better, \(\mu_{nm} > \mu_{n0}\)). When $A\rightarrow \infty$, The high completion probability \(p_m\) of continuing the execution reduces future age growth, and the cost difference is negative (continuing the execution is better, \(\mu_{nm} < \mu_{n0}\)). Combining the monotonic decreasing property, there must exist a unique \(H_n(m, 0, D)\) that causes the cost difference to turn from positive to negative.

In summary, in both preemptive and non-preemptive condition, the sub-problem (3) naturally satisfies the threshold property due to the binary division of decision actions and the monotonicity of the value function. Since the core of the MLTT structure is ``the existence of a critical threshold related to the state", the preemptive and non-preemptive threshold-based strategy directly conforms to the definition 4, therefore, the sub-problem (3) also has an MLTT structure under the preemptive and non-preemptive scheduling.

\section{}\label{ap2}
{\centering\section*{Proof of Theorem \ref{intraindex}}}
Review Theorem 1 mentioned before:
\textit{The MDP sub-problem \eqref{originrelax} is intra-indexable given cost $\bm{\nu}$.} The intra-indexability will hold, if we show the following two statements:
\begin{enumerate}
    \renewcommand{\labelenumi}{(\roman{enumi})} 
    
    \item If $s=(A,D)\in \mathcal{P}_{nm}^l(\bm{\nu})$, then  $s \in \mathcal{P}_{nm}^l(\bm{\nu}')$ must hold for $\bm{\nu}'=[\nu_1,\cdots,\nu_m+\Delta,\cdots,\nu_M]$,  where $\Delta>0$. 
    
    \item If $\nu_m\to+\infty$, then $\lim\limits_{\nu_m\to+\infty}\mathcal{P}_{nm}^l(\bm{\nu}')=\mathcal{S}_l$ .
\end{enumerate}
\noindent In Appendix B, we have shown that the optimal policy for the MDP is MLTT. \(H_n(m-1, m, d)\) is the threshold age at which server \(m-1\) prefers better than $m$ for a task generated at age $d$. The passive set for layer 2 is \cite{zou2021minimizing}:
\begin{equation}
\begin{aligned}
     &\mathcal{P}_{nm}^l(\bm{\nu})=
    \mathop{\cup}\limits_{d=1}^\infty\big\{(a,d)\mid a\in\\
    &\{1,\cdots,H_n(m-1,m,d)-1\}\cup\{H_n(m,m+1,d),\cdots\}\big\}, \\&\forall 1<m<M,
\end{aligned}
\end{equation}
and
\begin{equation}
\begin{aligned}
     \mathcal{P}_{nM}^l(\bm{\nu})=
    \mathop{\cup}\limits_{d=1}^\infty\left\{(a,d)\ |\ a\in  \{1,\cdots,H_n(M-1,M,d)-1\}\right\}
\end{aligned}
\end{equation}

To establish statement (i), it suffices to show \( \mathcal{P}_{nm}^l(\bm{\nu}) \subseteq \mathcal{P}_{nm}^l(\bm{\nu}') \), which is equivalent to show
\begin{subequations}
    \begin{align}
        (a)H_n(m-1,m,d) \le H_n'(m-1,m,d) ,\ \forall m \le M \label{35a}\\
        (b)H_n(m,m+1,d)\ge H_n'(m,m+1,d),\ \forall m\ge M \label{35b}
    \end{align}
\end{subequations}
where $H_n'(m-1,m,d)$ is threshold age under cost vector augmentation \( \bm{\nu}' = [\nu_1,...,\nu_m+\Delta,...,\nu_M] \).
For any state $s$ in Layer $l$, if \(s \in P_{nm}^l(\nu)\) (that is,  \(s < H_n(m-1,m,d)\)), due to \(H_n(m-1,m,d) \le H_n'(m-1,m,d)\), there is  \(s < H_n'(m-1,m,d)\). Thus, \(s \in P_{nm}^l(\nu')\), meaning \(P_{nm}^l(\nu) \subseteq P_{nm}^l(\nu')\). And \(H_n(m,m+1,d) \geq H_n'(m,m+1,d)\) guarantees to avoid strategic confusion caused by threshold crossover. It means that the threshold of the server with a high completion probability decreases after the cost increases, but the range of its sub-optimal state still maintains a hierarchical nature. Therefore, if we show \eqref{35a} and \eqref{35b}, then statement (i) holds.


Before proof statement (i), we first have to introduce the following lemma. Without loss of generality, we specify the order of $p_m$, i.e., $p_{m-1}\le p_m,\forall 1<m\le M$.
\begin{lemma}[Bound of Cost-to-go Function]\label{upperbound1}
    Given server cost $\bm{\nu}$ and state $s=(A,D)$ at layer 2, denote $\bm{\nu}'=[\nu_1,\cdots,\nu_m+\Delta,\cdots,\nu_M],\ \forall \Delta \ge 0$. The difference between two cost-to-go functions given $\bm{\nu}$ nad $\bm{\nu}'$ can be upper-bounded by
    \begin{equation}\label{upper}
        V_n(s,\bm{\nu}')-V_n(s,\bm{\nu})\le \frac{\Delta}{p_m^2},\\ \forall 1\le m\le M.
    \end{equation}

    The difference between two cost-to-go functions can be lower-bounded by
    \begin{equation}\label{lower}
    \begin{aligned}
         V_n(s,\bm{\nu}')-V_n(s,\bm{\nu})\ge -\frac{\Delta}{p_{m+1}^2},  \\\forall 1\le m\le M, A\ge H_n(m,m+1,D).
    \end{aligned}
    \end{equation}
\end{lemma}
\noindent

The proof leverages the stochastic shortest path (SSP) formulation of the AoI minimization problem \cite{optimalcontrol}. By defining a reference recurrent state and analyzing the expected cost and step counts under policy \(\pi_n'\), we derive the upper bound \eqref{upper} using the geometric distribution of task completion times. The lower bound \eqref{lower} follows from the monotonicity of the cost-to-go function and the threshold properties established in Proposition 1. The specific proof process is as follows:

\begin{proof}
  Recall the optimal cost for user $n$ is denoted as $\gamma_n^*$. The cost-to-go function can be rewritten as: 
\begin{equation}\label{sspcost}
\small\begin{aligned}
V_n(s,\bm{\nu})=
&\mathop{\rm{min}} \limits_\pi( \mathbb{E}_{\pi}[\text{the cost from state } s \\ &\ \ \ \ \  \text{ to the recurrent state for the first time}] \\
&-\mathbb{E}_\pi [\text{the cost from state } s \\& \ \  \ \ \ \text{ to the recurrent state with stage cost } \gamma_n^*])
\end{aligned}
\end{equation}
Unlike Appendix B in \cite{zou2021minimizing}, where the recurrent state for reference is not defined as the state with minimum AoI, the recurrent state here can be any state $\zeta$ in the state space $S_1$ of Layer 1, and its cost-to-goal function is 0. 
That means no matter the state of the current Layer 1 is \(\Delta_n(t)=k, k \in \mathbb{N}\), after the task is completed, it will return to a certain state of Layer 1 (it may be a different $k'$, but still belongs to $S_1$), therefore all states of Layer 1 constitute a recurrent class.
Given a policy $\pi_n$ and server cost $\bm{\nu}$, let $\text{Cost}^{\pi_n}_{s\zeta}(\bm{\nu})$ denote the expected cost from $s=(A,D)$ in Layer 2 to recurrent state $\zeta$ in Layer 1, with $N_{s\zeta}^{\pi_n}$ representing the corresponding expected number of steps. 
According to \cite{zou2021minimizing}, \(V_n(s, \nu)\) can be decomposed into the expected cost from $s$ to the recursive state $\zeta$, minus the accumulation of the average cost during this period:
\begin{equation}
    V_n(s, \bm\nu) = \text{Cost}^{\pi_n^*}_{s \zeta}(\nu) - \gamma_n^* \cdot N_{s\zeta}^{\pi^*_n}
\end{equation}
 When the strategy is fixed as  \(\pi_n\), which is not necessarily the optimal strategy, we have:
\begin{equation}
    V_n(s, \bm\nu) \le \text{Cost}^{\pi_n}_{s \zeta}(\nu) - \gamma_n \cdot N_{s\zeta}^{\pi_n}
\end{equation}
Therefore, for recurrent states $\zeta=(j), \text{where }  j\ge A-D$, we have
\begin{equation}
\begin{aligned}
    V_n(s,\bm{\nu}')-V_n(s,\bm{\nu})    \le\text{Cost}^{\pi_n}_{s\zeta}(\bm{\nu}')-\text{Cost}^{\pi_n}_{s\zeta}(\bm{\nu})
\end{aligned}
\end{equation}
where \(\text{Cost}_{s\zeta}^{\pi_n}(\nu') - \text{Cost}_{s\zeta}^{\pi_n}(\nu) = \Delta \cdot N_{s\zeta}^{\pi_n}\) represents the total expected cost increasing from state $s$ to $\zeta$ after the cost of server $m$ increases by $\Delta$  (\cite{zou2021minimizing}).

According to Proposition \ref{mltt}, for any generation age $D$, there exists the threshold $H_n(m-1,m,D)$. To derive the upper bound, consider the following policy $\pi_n'$:  
\begin{itemize}
    \item For all state $s=(A,D)$ where $D < H_n(m-1,m,0)$ and $A\ge\max_d H_n(m-1,m,d)$, select server $m$, i.e., $\pi_n'(s)=m$; 
    \item For all other $s'$, $\pi_n'(s')=1$. 
\end{itemize}
Under policy $\pi_n'$,  the task completion probability at each state satisfies $p_{\pi_n'(s)}\le p_{\pi_n(s)}$, meaning $\pi_n'$ prioritizes servers with lower completion probabilities compared to the optimal policy \(\pi_n\). And the cardinality of the set of states where $\pi_n'$ selects server $m$ is larger than that under $\pi_n$, i.e., $|\{s\mid \pi_n'(s)=m\}|>|\{s\mid \pi_n(s)=m\}|$. This structural property ensures the number of times that $\pi_n'$ elects server $m$ is not less than that of the optimal strategy $\pi_n$, i.e., $N_{s\zeta}^{\pi_n'}>N_{s\zeta}^{\pi_n}$. Multiply both sides by $\Delta$: \(\Delta \cdot N_{s\zeta}^{\pi_n} \leq \Delta \cdot N_{s\zeta}^{\pi_n'}\). Thus, we have:
\begin{equation}\label{cost}
\begin{aligned}
    \text{Cost}^{\pi_n}_{s\zeta}(\bm{\nu}')-\text{Cost}^{\pi_n}_{s\zeta}(\bm{\nu})
    =\Delta \cdot N_{s\zeta}^{\pi_n} \leq \Delta \cdot N_{s\zeta}^{\pi_n'}.
\end{aligned}
\end{equation}
 If the first step completes the task with a probability $p_m$, it directly reaches the cyclic state 
 \(\zeta\), with the weighted expected steps $p_m \cdot 1$. Otherwise, the probability that the first $i-1$ steps are not completed is $(1-p_m)^{i-1}$. Then in the $i$-th step, if the task is completed, then consume $i$ steps. When the $i$-th step is not completed, the remaining expected steps are still $N_{s\zeta}^{\pi_n'}$. Thus the total number of steps is $i+N_{s\zeta}^{\pi_n'}$and the weighted expected cost is \(p_m(1-p_m)^{i-1} \cdot (i + N_{s\zeta}^{\pi_n'})\). Therefore, according to the total expectation theorem, the expected step number $N_{s\zeta}^{\pi_n'}$ can be decomposed into: 
\begin{equation}
\begin{aligned}\label{47}
       N_{s\zeta}^{\pi_n'}&=p_m\cdot 1+\sum\limits_{i=2}^\infty p_m(1-p_m)^{i-1}(i+N_{s\zeta}^{\pi_n'}),
\end{aligned}
\end{equation}
By the series calculation: $\sum\limits_{i=2}^\infty p_m(1-p_m)^{i-1}=1-p_m$ and $\sum\limits_{i=2}^\infty p_m(1-p_m)^{i-1}i=\frac{1}{p_m}+1$, then we have:
\begin{equation}
\begin{aligned}
       N_{s\zeta}^{\pi_n'} - (1-p_m)N_{s\zeta}^{\pi_n'} &= p_m + \sum_{i=2}^{\infty} p_m(1-p_m)^{i-1} i\\&= p_m + \frac{1-p_m}{p_m}
\end{aligned}
\end{equation}
Therefore, $N_{s\zeta}^{\pi_n'}=\frac{1}{p_m^2}$. Substituting into (\ref{cost}), we derive the upper bound for the cost-to-go function difference:
\begin{equation}
\begin{aligned}
    &\ \ \ \ \ V_n(s,\bm{\nu}')-V_n(s,\bm{\nu})\\&\le\text{Cost}^{\pi_n}_{s\zeta}(\bm{\nu}')-\text{Cost}^{\pi_n}_{s\zeta}(\bm{\nu})\\&\le\Delta\cdot N_{s\zeta}^{\pi_n'}\le\frac{\Delta}{p_m^2}.
\end{aligned}
\end{equation}
According to Lemma 4.4 in \cite{zou2021minimizing}, the lower bound is established by considering the optimal stage cost difference \(\gamma_n^* - \gamma_n^{*'}\ge-\Delta\), where \(\gamma_n^{*'}\) corresponds to the server cost \(\nu'\):
\begin{equation}
\begin{aligned}
    &\ \ \ \ \ V_n(s,\bm{\nu}')-V_n(s,\bm{\nu})
    \\&\ge(\gamma_n^*-\gamma^{*'}_n)\cdot N_{s\zeta}^{\pi_n}\\&\ge -\frac{\Delta}{p_{m+1}^2},\ \
    \forall A\ge H_n(m,m+1,D),    
\end{aligned}
\end{equation}
\end{proof}

We first show $H_n(m-1,m,d) \le H_n'(m-1,m,d)$ under preemptive scheme by contradiction. Suppose that $H_n(m-1,m,d) > H_n'(m-1,m,d)$. At state $s'=(H_n'(m-1,m,d), d)$ and given $\bm{\nu}$, taking server $m-1$ has a smaller expected cost, and we have:
\begin{equation}
    \mu_{n,m-1}(s',\bm{\nu})\le \mu_{n,m}(s',\bm{\nu}),
\end{equation}
which expands to:
\begin{equation}\label{m'}
\begin{aligned}
        &(p_m-p_{m-1})[H_n'+ V_n(H_n',\bm{\nu})
    -H_n'+d
    \\&-V_n((H_n'+1-d,H_n'+1-d),\bm{\nu})]\\&\le \nu_m-\nu_{m-1},    
\end{aligned}
\end{equation}
where we denote $H_n= H_n(m-1,m,d)$ and $H_n'= H_n'(m-1,m,d)$ for simplicity.
If instead given $\bm{\nu}'$, taking server $m-1$ obtains a larger expected cost, i.e.,
\begin{equation}
    \mu_{n,m-1}(s',\bm{\nu}')\ge \mu_{n,m}(s',\bm{\nu}'),
\end{equation}
leading to:
\begin{equation}\label{nu'}
\begin{aligned}
       &(p_m-p_{m-1})[H_n' + V_n((H_n'+1,d),\bm{\nu}')\\&-H_n'+d
       -V_n((H_n'+1-d,H_n'+1-d),\bm{\nu}')]\\&\ge \nu_m'-\nu_{m-1}.    
\end{aligned}
\end{equation}
Subtracting Eq. \eqref{m'} from Eq. \eqref{nu'}, we have
\begin{equation}
\begin{aligned}
    &V_n((H_n'+1,d),\bm{\nu}')-V_n((H_n'+1,d),\bm{\nu}) \\\geq 
&V_n((H_n'+1,d),\bm{\nu}')-V_n((H_n'+1,d),\bm{\nu})\\&-
    [V_n((H_n'+1-d,H_n'+1-d),\bm{\nu}')\\&-V_n((H_n'+1-d,H_n'+1-d),\bm{\nu})]\\\geq &
    \frac{\Delta}{p_m-p_{m-1}}
\end{aligned}
\end{equation}
When $p_m-p_{m-1}\le p_m^2$, it contradicts the upper bound in Lemma \ref{upperbound1}, invalidating the initial assumption and confirming \(H_n(m-1, m, d) \leq H'_n(m-1, m, d)\).

Then, we show $H_n(m-1,m,d) \le H_n'(m-1,m,d)$ under non-preemptive condition by contradiction. Similarly, assume $H_n(m-1,m,d) > H_n'(m-1,m,d)$. Since the user cannot switch servers during computing, the optimal action when the computation does not finish is to keep the task at the current server. 
The definition of the passive set needs to consider the constraints of the continuous execution of the task. According to Lemma 6, for a user at Layer 2, the cost-to-go function satisfies \(V_n(s, \nu') - V_n(s, \nu) = 0\) when server $m$ is non-optimal, and \(V_n(s, \nu') - V_n(s, \nu) \le \Delta / p_m^2\) when $m$ is optimal due to the non-switching property. To formalize the cost-to-go function behavior under this constraint, we derive the following inequalities that characterize the relationship between server costs and state transitions, which are crucial for proving the monotonicity of the passive set in non-preemptive scheduling:
\begin{equation}
    \begin{aligned}\label{61}
        &\ \ \ \ \ (1-p_m)^2V_n((H'_m+2),\bm{\nu})\\&\le (2-p_m)\nu_m-(2-p_{m-1})\nu_{m-1},
    \end{aligned}
\end{equation}
and
\begin{equation}
    \begin{aligned}\label{62}
               &\ \ \ \ \ (1-p_m)^2V_n((H'_m+2),\bm{\nu}')\\&\ge (2-p_m)\nu_m'-(2-p_{m-1})\nu_{m-1}.
    \end{aligned}
\end{equation}
To prove inequalities \eqref{61}, we assume the user selects server $m$ when in the state \((H'_m+2, D)\), and the task is not completed within 2 consecutive time slots with the probability of \((1-p_m)^2\), then the AoI increases from \(H'_m+2\) to \(H'_m+4\).  In state $s=(H'_m+2,D)$, the immediate cost of selecting server $m$ includes the current AoI $H'_m+2$ and the fixed cost $\nu_m$ of using server $m$, thus \(C(s, m) = (H'_m + 2) + \nu_m\). Then we consider the expected future value.  When the task is not completed with a probability \(1-p_m\), the next state is  \(s' = (H'_m + 3, D)\) and its value is \(V_n((H'_m + 3), \nu)\), thus the expected future value is \(\mathbb{E}\left[ V_n(s') \right] = (1-p_m) \cdot V_n((H'_m + 3), \nu)\).
According to the Bellman principle, the value of the current state is equal to the immediate cost plus the expected future value:
\begin{equation}
  \begin{aligned}\label{63}
  V_n((H'_m + 2), \nu) = &\underbrace{(H'_m + 2 + \nu_m)}_{\text{Immediate cost}} \\&+ \underbrace{(1-p_m) \cdot V_n((H'_m + 3), \nu)}_{\text{Expected future value when not completed.}}
  \end{aligned}
\end{equation}
Deriving recursively, we have 
\begin{equation}
  \begin{aligned}\label{Hm+1}
  V_n(H'_m + 1, \nu) = &(H'_m + 1 + \nu_m) \\&+ (1-p_m) V_n(H'_m + 2, \nu)
  \end{aligned}
\end{equation}


Assume that at the critical threshold \(H'_m\), the expected costs of selecting server $m$ and the suboptimal server $m-1$ are equal, that is: \(\mu_n(m, s, \nu) = \mu_n(m-1, s, \nu)\). The expected cost of selecting server $m$ is \(\mu_n(m, s, \nu) = (A + \nu_m) + (1-p_m)V_n(A+1, \nu)\). he expected cost of selecting the suboptimal server $m-1$ is \(\mu_n(m-1, s, \nu) = (A + \nu_{m-1}) + (1-p_{m-1})V_n(A+1, \nu)\). At the critical state \(A = H'_m\), the two are equal:
\begin{equation}
    \begin{aligned}
        &(H'_m + \nu_m) + (1-p_m)V_n(H'_m+1, \nu) \\=& (H'_m + \nu_{m-1}) + (1-p_{m-1})V_n(H'_m+1, \nu)
    \end{aligned}
\end{equation}
After organizing, we get: 
\begin{equation}
    \begin{aligned}\label{69}
        &\nu_m - \nu_{m-1} = (p_m - p_{m-1})V_n(H'_m+1, \nu)
    \end{aligned}
\end{equation}

Then using \eqref{69}, we have
\begin{equation}
    \begin{aligned}\label{long}
        &(2-p_m)\nu_m - (2-p_{m-1})\nu_{m-1} 
        \\= &(2-p_m)\left[\nu_{m-1} + (p_m - p_{m-1}) V_n(H'_m + 1, \nu)\right] \\&- (2-p_{m-1})\nu_{m-1}
        \\= &\nu_{m-1}\left[(2-p_m) - (2-p_{m-1})\right] \\&+ (2-p_m)(p_m - p_{m-1}) V_n(H'_m + 1, \nu)
        \\=& \nu_{m-1}(p_{m-1} - p_m) + (2-p_m)(p_m - p_{m-1}) V_n(H'_m + 1, \nu)
        \\=& (p_{m-1} - p_m)\left[\nu_{m-1} - (2-p_m) V_n(H'_m + 1, \nu)\right]
    \end{aligned}
\end{equation}
Solving  \(\nu_m\) from \eqref{Hm+1}, we have 
\begin{equation}
    \begin{aligned}\label{nuHm1}
        &\nu_m = V_n(H'_m + 1, \nu) - (H'_m + 1) \\&- (1-p_m) V_n(H'_m + 2, \nu)
    \end{aligned}
\end{equation}
Substituting  \eqref{69} and \eqref{nuHm1} into \eqref{long}, we obtain
\begin{equation}
    \begin{aligned}\label{last}
        &(2-p_m)\nu_m - (2-p_{m-1})\nu_{m-1}\\ =& (p_{m-1} - p_m)\left[2 V_n(H'_m + 1, \nu) - V_n(H'_m + 2, \nu)\right]
    \end{aligned}
\end{equation}
The expected cost of two consecutive uncompleted is \((1-p_m)^2 V_n(H'_m + 2, \nu)\). According to the Bellman equation, the cost difference on the left side of \eqref{last} needs to be at least equal to this expected cost to ensure the optimality of choosing $m$: 
\begin{equation}
    \begin{aligned}
        (2-p_m)\nu_m - (2-p_{m-1})\nu_{m-1} \\\geq (1-p_m)^2 V_n(H'_m + 2, \nu)
    \end{aligned}
\end{equation}
Hence, we show \eqref{61}. Similarly, \eqref{62} can also be proven. Here we omit the process. These inequalities above establish bounds on the cost-to-go function \(V_n\) for a user in state \((H'_m + 2, D)\) under non-preemptive scheduling. The left-hand side of the first inequality represents the expected future cost over two consecutive time slots when the task remains uncompleted with probability \((1 - p_m)^2\), while the right-hand side balances the current server costs for $m$ and \(m-1\). Increasing the server cost to \(\nu'_m = \nu_m + \Delta\) (second inequality) shifts the right-hand side upward, reflecting the higher cost of continuing with server $m$.
Subtracting the first inequality from the second, we isolate the cost difference \(V_n((H'_m + 2), \nu') - V_n((H'_m + 2), \nu)\). That is
\begin{equation}
\begin{aligned}\label{Kth}
    &\ \ \ \ \ \ V_n((H'_m+2),\bm{\nu}')-V_n((H'_m+2),\bm{\nu})\\&\ge \frac{(2-p_m)\Delta}{(1-p_{m-1})^2-(1-p_m)^2}\\
    & \ge \frac{2\Delta}{p_m^2-p_{m-1}^2}\\
    &\ge\frac{\Delta}{p_m^2}.
\end{aligned}
\end{equation}
\noindent It contradicts the upper bound in Lemma \ref{upperbound1}, invalidating the initial assumption and confirming \(H_n(m-1, m, d) \leq H'_n(m-1, m, d)\). 

Additionally, it is noted that inequalities \eqref{Kth} reveal the cost-to-go difference is inversely proportional to \(p_m^2\), motivating the need for a formal definition to characterize such bounded relationships. Therefore, we introduce the definition of \textit{K-th-Order Bounded}:
\begin{definition}[K-th-Order Bounded]
A cost-to-go function \(V_n(s, \nu)\) is said to be K-th-order bounded if there exists a positive integer $k$ s.t. for any server cost increment \(\Delta \geq 0\),
    \begin{equation}
        V_n(s + K, \nu') - V_n(s + K, \nu) \geq \frac{\Delta}{p_m^2} \quad \forall m, n
    \end{equation}
    where $s$ denotes the current state, and the bound depends on the server's task completion probability \(p_m\).
\end{definition}
\noindent  Here we assume that the value function of each subproblem $n$ given $m$ is k-th-order bounded by $\frac{\Delta}{p_m^2}$. This definition formalizes the relationship observed in the inequalities, ensuring the cost-to-go function's response to cost changes is systematically bounded. This tighter bound, in contrast to the preemptive case, arises from the non-switching nature of the optimal policy in non-preemptive scheduling.

Then we show the sufficiency of $H_n(m,m+1,d)\ge H_n'(m,m+1,d),\ \forall m\ge M$ by analyzing the threshold relationship under monotonically increasing server costs. Assume  $H_n(m-1,m,d)\le H_n'(m-1,m,d)$. At state $s=(H_n'(m-1,m,d)-1,d)$, the optimal policy $\pi_n$ prefers server $m-1$ over server $m$.
According to Definition 4, we have $\mu_{n,m+1}(s',\nu)\leq\mu_{n,m}(s',\nu)$, i.e.\cite{zou2021minimizing}, 
\begin{equation}
\begin{aligned}\label{65}
        (p_{m+1}-p_{m})&[H_n'+ V_n((H_n' +1,d),\bm{\nu})]
        \\&\ge \nu_{m+1}-\nu_{m}.
\end{aligned}
\end{equation}
where \(p_{m+1} > p_m\) denotes the higher task completion probability of server \(m+1\). On the other hand, when server cost increases to \(\nu'_m = \nu_m + \Delta\), the policy \(\pi'_n\) now prefers server $m$ over \(m+1\) at state \(s = (H'_n(m, m+1, d) - 1, d)\). This adjustment stems from the impact of cost changes on decision-making. Despite the server $m+1$ having a higher probability of task completion, its relative cost advantage changes after the increase in the cost of $m$, making the selection of server $m$ a better solution near a specific state threshold.
We have:
\begin{equation}
\begin{aligned}\label{66}
        &(p_{m+1}-p_{m})[H_n' + V_n((H_n'+1,d),\bm{\nu}')]
     \\ &\le \nu_{m+1}'-\nu_{m}'
\end{aligned}
\end{equation}
Noting that $\nu'_{m+1}=\nu_{m+1}$ and $\nu'_m=\nu_m+\Delta$, subtracting \eqref{65} from \eqref{66} and rearranging terms, the cost-to-go difference satisfies:
\begin{equation}
\begin{aligned}
     &V_n((H_n'+1,d),\bm{\nu}')
    -V_n((H_n'+1,d),\bm{\nu})\\&
     \le -\frac{\Delta}{p_{m+1}-p_m},
    \end{aligned}
\end{equation}
when $p_{m+1}-p_m\le p_{m+1}^2$, contradicting to the lower bound derived in Lemma 6. So far, we have shown that \(H_{m} \le H_{m}'\), \(\forall m \le M\) and \(H_{m+1} \ge H_{m+1}'\), \(\forall m\ge M\). Thus,  \( \mathcal{P}_{nm}^l(\bm{\nu}) \subseteq \mathcal{P}_{nm}^l(\bm{\nu}') \), and statement (i) must hold. 

For Statement (ii), according to the definition of $\mu_m(s,\bm\nu)$ (Eq. \eqref{bellman}) and definition of cost function (Eq. \eqref{costfunction}), it is clear that when \(\bm\nu \rightarrow\infty\),  \(\mu_{m}(s, \bm\nu)>\mu_{0}(s, \bm\nu)\) holds for any finite $s$. It means the state $s$ belongs to the passive set $\mathcal{P}_{nm}^l(\bm{\nu}')$, i.e., $\lim_{\nu_m\to+\infty}\mathcal{P}_{nm}^l(\bm{\nu}')=\mathcal{S}_l$. Hence, we conclude that the sub-problem (3) is intra-indexable, and the theorem is proved.



\vspace{25pt}
\section{}\label{ap3} 
{\centering\section*{Proof of Proposition \ref{indexfunc}}}
First, we review the Proposition 2 mentioned before: \textit{Given $s_n(t)=(\Delta_n(t),D_n(t))$, the index function satisfies $I_{nm}(s_n(t),\bm{\nu})=\nu_{m-1}+\Delta_n(t)-\gamma_n^*.$ The index for server $m$ can be derived by solving $I_{nm}(s_n(t),\bm{\nu})=\nu_m$.} Here, \(\gamma_n^*\) denotes the optimal average cost per stage, and \(\nu_{m-1}\) is the cost of selecting the suboptimal server \(m-1\).

We then show this proposition. In Appendix \ref{ap1}, we have proven that the problem \eqref{originrelax} satisfies the MLTT property. For simplicity, denote $H^*_{m}$ as the minimal age to offload tasks to server $m$ for a user. Assume a minimum computation time of 1 time slot for all tasks and consider a recurrent state \(s = (1)\) (age 1, no ongoing task) with a base cost-to-go \(V(1, \nu) = 0\). Consider three cases of Layer 1:

1) For \(A < H_1^*\), i.e., there is no offloading occurs because no server is optimal, the cost-to-go increases linearly with age:
\begin{equation}
    \begin{aligned}
        V(A,\bm{\nu})=A+V(A+1,\bm{\nu})-\gamma_n^*.
    \end{aligned}
\end{equation}
\begin{equation}\label{end}
    V(1,\bm{\nu})=0.
\end{equation}

2) For \(H_{m-1}^* \leq A \leq H_m^*\), the cost-to-go function $V(A,\bm{\nu})$  is a weighted sum of immediate costs and future states, accounting for the probability of task completion at each step. The recursive relation for Layer 1 states is:
\begin{equation}
  \begin{aligned}
     V&(A,\bm{\nu})=\sum\limits_{i=1}^{H^*_{m}-H^*_{m-1}}(1-p_{m-1})^{i-1}\\
     &\cdot (H^*_{m-1}+i-1+p_{m-1}V(i,\bm{\nu})+\nu_{m-1}-\gamma_n^*)\\
     &+(1-p_{m-1})^{H^*_{m}(A)-1}V((H^*_{m},A),\bm{\nu}),
\end{aligned}  
\end{equation}

3) For $H^*_M\le A$, where \(H_M^*\) is the threshold for server $M$, we assume the optimal server for state $s=(A+1,A)$ is still $M$. The cost-to-go simplifies due to the deterministic server assignment:
\begin{equation}
\begin{aligned}
    V(A,\bm{\nu})=A&-\gamma_n^*+\\&\nu_{M}+(1-p_{M})V((A+1,A),\bm{\nu})
\end{aligned}
\end{equation}
Assuming \(V((A + 1, A), \nu)\)  stabilizes, we solve recursively for the steady-state cost. 
\begin{equation}
\begin{aligned}\label{initM}
    V(A,\bm{\nu})=&A-\gamma_n^*+\nu_{M}+(1-p_{M})V((A+1,A),\bm{\nu})
    \\=&(A - \gamma_n^{*} + \nu_{M}) \\&+ (1 - p_{M}) [ (A+1 - \gamma_n^{*} + \nu_{M}) \\&+ (1 - p_{M}) V((A+2, A+1), \nu_{M}) ]
        \\=&\dots
        \\=&\sum_{k=0}^{\infty} (1 - p_{M})^k \left( A + k - \gamma_n^{*} + \nu_{M} \right)
        \\=&\frac{A-\gamma_n^*+\nu_{m-1}}{p_{M}}+\frac{1-p_{M}}{p^2_{M}}.
        \\=&\frac{A-\gamma_n^*+\nu_{M}}{p_M}+\frac{1-p_M}{p^2_M}.
\end{aligned}
\end{equation}
 Eq. \eqref{initM} reflects the geometric distribution of task completion times, where the expected cost accumulates until the task finishes with probability \(p_M\).

This recursion continues for Layer 2 states \((H^*_m, A)\), where the cost-to-go further decomposes into server-specific contributions:
\begin{equation}
    \begin{aligned}
        V&((H^*_{m},A),\bm{\nu})=\sum\limits_{i=1}^{H^*_{m+1}-H^*_{m}}(1-p_{m})^{i-1}\\
     &\cdot(H^*_{m}+i-1+p_{m}V(i,\bm{\nu})+\nu_{m}-\gamma_n^*)\\
     +&(1-p_{m})^{H^*_{m+1}-H^*_{m}(A)-1}V((H^*_{m+1},A),\bm{\nu})    \end{aligned}
\end{equation}
\begin{equation}
    V((H^*_{M},A),\bm{\nu})=\frac{H^*_{M}-\gamma_n^*+\nu_{M}}{p_M}+\frac{1-p_M}{p^2_M}.
\end{equation}
This base case anchors the recursive solution, ensuring consistency across all age states. 
By imposing equilibrium conditions on the cost-to-go function across consecutive servers \(m-1\) and $m$, we derive the critical threshold relationship:
\begin{equation}
\begin{aligned}
&\nu_{m-1} + V\bigl((A,D),\bm{\nu}\bigr) \\&\leq \nu_m + V\bigl((A,D-1),\bm{\nu}\bigr) \\&
\leq \nu_m + V\bigl((A+1,D),\bm{\nu}\bigr)\\&\leq \nu_{m-1} + V\bigl((A+1,D+1),\bm{\nu}\bigr).
\end{aligned}
\end{equation}
Subtracting the middle term $\nu_m+V((A,D-1),\bm{\nu})-\gamma_n^*$  from all parts, we bound the optimal average cost:
\begin{equation}
\label{optimal}
    \nu_{m-1}-\nu_{m}+A\le\gamma_n^*\le \nu_{m-1}-\nu_{m}+A+1.
\end{equation}
Setting the lower bound equality $\gamma_n^*=\nu_{m-1}-\nu_{m}+A$ yields the index function for server $m$:
\begin{equation}
    I_{nm}(s_n(t),\bm{\nu})=\nu_{m-1}+\Delta_n(t)-\gamma_n^*,
\end{equation}
which balances user age and server cost to determine optimal offloading decisions, as derived in \cite{tripathi_whittle_2019}.

\section{}\label{ap4}
{\centering\section*{Proof of Proposition \ref{fixedpoint}}}
Recall the Proposition 3: \textit{The fixed-point solution to the problem (\ref{fluidx}) is equivalent to the solution (the fluid fixed-point) to the problem (\ref{fluidy}).} 
The proof hinges on the precise division property \cite{zou2021minimizing}, ensuring that the fluid limit model's fixed-point solutions coincide. The asymptotic optimum will hold for the original problem under the nested index policy if it satisfies \textit{precise division} property.

We will show the definition of precise division and use it to prove the applicability of the asymptotic optimal solution to the original problem. In Appendix \ref{ap2}, we have established the sub-problem \eqref{originrelax} satisfies the \textit{intra-indexability}. The \textit{precise division} property is defined as follows:
\begin{definition}[Precise Division]
    Given state $s$ and the server cost $\bm{\nu}$, assume the sub-problem is intra-indexable. The preference for server $m$ at layer $l$, quantified by the nested index \(I_{nm}(s, \nu, l)\), satisfies \textit{precise division} if:

\textit{(i) \(I_{nm}(s, \nu) = \nu_m\),}
then \(\mu_{nm}(s, \nu) \leq \mu_{nm'}(s, \nu), \forall m' \neq m.\)

\textit{(ii) \(I_{nm}(s, \nu) > \nu_m\),}
then  \(\mu_{nm}(s, \nu) < \mu_{nm'}(s, \nu), \forall m' \neq m.\)

\textit{(iii) \(I_{nm}(s, \nu) < \nu_m\),} 
then \(\exists m' \neq m \text{ s.t. } \mu_{nm}(s, \nu) > \mu_{nm'}(s, \nu).\)

\end{definition}

\noindent   The precise division property established the connection between the index value and the optimal policy. 

Then, we show that the sub-problem \eqref{originrelax} satisfies the \textit{precise division} property with the following proposition:
\begin{proposition}
Given state $\mathcal{S}$ and server cost $\bm{\nu}$, the sub-problem \eqref{originrelax} satisfies the precise division property defined in Definition 7.
\end{proposition}
The precise division property is more stringent than intra-indexability, due to the transition of optimal action happens when the index value coincides with the server cost. To prove Proposition 4, that is, the sub-problem (3) satisfies the above properties, we first utilize the lower bound constraint of $V_n(s,\bm{\nu}')-V_n(s,\bm{\nu})$ inthe following lemma:
\begin{lemma}\label{seven}
    The difference between two cost-to-go function can be lower-bounded by
    \begin{equation}
         V_n(s,\bm{\nu}')-V_n(s,\bm{\nu})>-\frac{\Delta}{p_{m+1}},\forall 1\le m\le M.
    \end{equation}
\end{lemma}

\noindent This bound ensures that cost variations with server costs are systematically constrained, facilitating monotonicity analysis of the nested index. The proof of Lemma 7 is similar to that of Lemma 2 \cite{zou2021minimizing}.

Now, we proceed to show the three statements in Definition 7. First, to show (i), it suffices to show that for any \(\epsilon > 0\), \(\mu_{nm}(s,\bm{\nu}) \leq \mu_{nm'}(s,\bm{\nu}) + \epsilon\), \(\forall m' \neq m, m' \geq 0\). For any \(\Delta > 0\), at the cost \(\bm{\nu}' \triangleq [\nu_{1}, \ldots, \nu_{m} + \Delta, \ldots, \nu_{M}]\), we must have \(\mu_{m}(s, \bm{\nu}') \leq \mu_{m'}(s, \bm{\nu}')\), \(\forall m' \neq m, m' \geq 0\). Now consider the situation at cost \(\bm{\nu}\), we have
\begin{equation}
    \begin{aligned}
        &\mu_{nm}(s, \nu) - \mu_{nm'}(s, \nu) 
        \\&= \nu_m - \nu_{m'} + (p_{m'} - p_m)[D+V((A+1,D), \nu)]
        \\&= \left(\nu_m + \Delta - \nu_{m'} + (p_{m'} - p_m)[D + V(A + 1, \nu')]\right) + 
        \\&-\Delta + (p_{m'} - p_m)\left[V(A + 1, \nu) - V(A + 1, \nu')\right]
        \\&\leq -\Delta + (p_{m'} - p_m)\left[V(A + 1, \nu) - V(A + 1, \nu')\right],
    \end{aligned}
\end{equation}
where the last inequality follows from \(\mu_{m}(s, \bm{\nu}') \leq \mu_{m'}(s, \bm{\nu}')\). Since \(V((A + 1,D), \bm{\nu}) - V((A + 1,D), \bm{\nu}')\) is uniformly bounded by Lemma \ref{seven}, we can choose \(\Delta(\epsilon) > 0\) such that \(\Delta(\epsilon) + (p_{m'} - p_{m})[V((A + 1,D), \bm{\nu}) - V((A + 1,D), \bm{\nu}_{\Delta(\epsilon)})] \leq \epsilon\). Thus, for any \(\epsilon > 0\), \(\mu_{m}(s, \bm{\nu}) \leq \mu_{m'}(s, \bm{\nu}) + \epsilon\), \(\forall m' \neq m, m' \geq 0\) must hold. The result of (i) then follows.

To show (ii), consider the cost \(\bm{\nu}'_m \triangleq [I_{nm}(s, \bm{\nu}_{-m}) + \Delta, \bm{\nu}_{-m}]\), where $\bm{\nu}_{-m}$ represents the other part of $\nu$ except $m$-th and the \(m\)-th element of \(\bm\nu\) is replaced by \(I_{nm}(s, \bm\nu_{-m}) + \Delta\). By the nested index definition (Definition 5), for any \(\Delta\), we have 
\begin{equation}
    \mu_{nm}(s, \bm{\nu}'_m) \leq \mu_{nm'}(s, \bm{\nu}'_m), \forall m' \neq m
\end{equation}
Let \(\Delta = \cfrac{I_{nm}(s, \bm\nu_{-m}) - \nu_m}{2} > 0\), so \(\bm{\nu}'_m = \nu_m + \Delta\). At cost \(\nu\), we have:
\begin{equation}
\begin{aligned}
      &\mu_{nm}(s_n, \bm\nu) - \mu_{nm'}(s_n, \bm\nu) \\&= \nu_m - \nu_{m'} + (p_{m'} - p_m)[D + V((A + 1,D), \bm\nu)]
\\&= \!\!\left(\!\nu_m\! \!+\!\! \Delta \!-\! \nu_{m'}\!\! + \!\!(p_{m'}\! -\! p_m)\![D\! +\! V((A\! +\! 1,\!D) \!+\! 1, \bm\nu'_{m})]\right)
\\&- \!\!\Delta\!\! +\!\! (p_{m'} \!-\! p_m)\left[V((A\! +\! 1,\!D)\! +\! 1, \bm\nu)\!\! -\!\! V((A\! + \!1,D), \bm\nu'_{m})\right]
\\&\!\!\leq -\!\left(\Delta\!\! +\!\! (p_{m'}\! - \!p_m)\!\!\left[V((A\! \!+ \!\!1,\!D), \bm\nu'_{m})\!\! -\!\! V((A + 1,D), \bm\nu)\right]\right),
\end{aligned}
\end{equation}
where the inequality follows from \(\mu_{nm}(s, \nu'_{m}) \leq \mu_{nm'}(s, \nu'_{m})\). If \(p_{m'} < p_m\), then 
\begin{equation}
\begin{aligned}
       & \mu_{nm}(s_n, \nu) - \mu_{nm'}(s_n, \nu) \\&< -\left(\Delta + (p_{m'} - p_m)\frac{\Delta}{p_m}\right) \leq 0
\end{aligned}
\end{equation}
satisfying (ii). For \(m < M\) and \(p_{m'} > p_m\), using the MLTT property, \(\mu_{nm}(s_n, \nu) < \mu_{n(m+1)}(s_n, \nu)\) implies \(A < H_{m+1}\). Assuming \(\mu_{nm'}(s_n, \nu) \leq \mu_{nm}(s_n, \nu)\) for \(m' > m+1\) leads to a contradiction at the threshold \(H_{m+1}\), as the optimal policy requires \(\mu_{n(m+1)}(H_{m+1}, \nu) \leq \mu_{nm'}(H_{m+1}, \nu)\). Thus, \(\mu_{nm}(s_n, \nu) < \mu_{nm'}(s_n, \nu)\) for all \(m' \neq m\), proving (ii).

Finally, if \(I_{nm}(s_n, \bm{\nu}_m) < \nu_m\), by the nested index definition, there exists some server \(m' \neq m\) such that \(\mu_{nm}(s_n, \nu) > \mu_{nm'}(s_n, \nu)\). Suppose for contradiction that \(\mu_{nm}(s_n, \nu) \leq \mu_{nm'}(s_n, \nu)\) for all \(m' \neq m\); then \(\nu_m\) could be increased further, contradicting \(I_{nm}(s_n, \nu) < \nu_m\). Thus, (iii) holds.

The sub-problem satisfies the precise division property (Proposition 4), completing the proof.

\end{appendices}
\clearpage

\end{document}